%% file: SHOQ-llncs.tex
\documentclass[11pt,envcountsame,envcountsect,runningheads]{llncs} 

\usepackage[a4paper,margin=2.52cm]{geometry}

\newcommand{\LongVersion}[1]{#1}
\newcommand{\ShortVersion}[1]{}
\newcommand{\VSpace}[1]{}

\newcommand{\myEnd}{\mbox{}\hfill$\Box$}

\usepackage{latexsym} 
\usepackage{xspace} 
\usepackage{amssymb} 
\usepackage{amsmath}
\usepackage{stmaryrd} 
\usepackage{url} 
\usepackage{oldgerm} 
\usepackage{graphicx} 
\input{xy} 
\xyoption{all} 
\usepackage[ruled,vlined,linesnumbered]{algorithm2e}

\input{SHOQ-def}

\LongVersion{\title{ExpTime Tableaux with Global Caching for the Description Logic SHOQ}}
\ShortVersion{\title{An ExpTime Tableau Method for Dealing with Nominals and Quantified Number Restrictions in Deciding the Description Logic SHOQ}} 

\titlerunning{ExpTime Tableaux for SHOQ}

\author{Linh Anh Nguyen\inst{1,2} \and Joanna Goli{\'n}ska-Pilarek\inst{3}} 
\institute{ 
Institute of Informatics, University of Warsaw\\ 
Banacha 2, 02-097 Warsaw, Poland\\ 
\email{nguyen@mimuw.edu.pl}\\
\and 
Faculty of Information Technology, 
VNU University of Engineering and Technology\\ 
144 Xuan Thuy, Hanoi, Vietnam\\
\and
Institute of Philosophy, University of Warsaw\\
Krakowskie Przedmie{\'s}cie 3, 00-927 Warsaw, Poland\\
\email{j.golinska@uw.edu.pl}\\[1em]
July, 2013 (last revised: July, 2014)
}

\authorrunning{L.A. Nguyen and J. Goli{\'n}ska-Pilarek}	 

\begin{document} 
\maketitle  
\sloppy 

\begin{abstract}
\ShortVersion{We present the first tableau method with an \EXPTIME (optimal) complexity for checking satisfiability of a knowledge base in the description logic \SHOQ, which extends \ALC with transitive roles, hierarchies of roles, nominals and quantified number restrictions. The complexity is measured using unary representation for numbers. Our procedure is based on global caching and integer linear feasibility checking.} 
\LongVersion{We give the first \EXPTIME (complexity-optimal) tableau decision procedure for checking satisfiability of a knowledge base in the description logic \SHOQ, which extends the basic description logic \ALC with transitive roles, hierarchies of roles, nominals and quantified number restrictions. The complexity is measured using unary representation for numbers. Our procedure is based on global caching and integer linear feasibility checking. 

\smallskip

}
\end{abstract}


\input{SHOQ-main}


\medskip

\noindent\textbf{Acknowledgments.} 
This work was supported by Polish National Science Centre (NCN) under Grants No.~2011/01/B/ST6/02759 (for the first author) and~2011/02/A/HS1/00395 (for the second author).


\LongVersion{\bibliography{modal}}
\ShortVersion{\bibliography{modal-short}}
\bibliographystyle{plain}


\input{SHOQ-proofs}


\end{document}

%% file: SHOQ-def.tex
\newcommand{\mand}{\sqcap}
\newcommand{\mor}{\sqcup}
\newcommand{\V}{\forall}
\newcommand{\E}{\exists}
\def\tuple#1{\langle#1\rangle} 

\newcommand{\ovl}[1]{\overline{#1}}

\newcommand{\EXPTIME}{\textsc{Exp\-Time}\xspace}
\newcommand{\NEXPTIME}{\textsc{NExp\-Time}\xspace}
\newcommand{\NtEXPTIME}{\textsc{N2Exp\-Time}\xspace}

\newcommand{\mT}{\mathcal{T}} 
\newcommand{\mR}{\mathcal{R}} 
\newcommand{\mA}{\mathcal{A}} 
\newcommand{\mI}{\mathcal{I}} 
\newcommand{\mE}{\mathcal{E}} 
\newcommand{\mS}{\mathcal{S}} 

\newcommand{\KB}{\mathit{KB}}

\newcommand{\nI}{\mathbb{I}} 
\newcommand{\nC}{\mathbb{C}} 
\newcommand{\nE}{\mathbb{E}}

\def\eqref#1{(\ref{#1})}

\newcommand{\koniec}{\mbox{}\hfill$\Box$}

\let\oldmarginpar\marginpar
\renewcommand\marginpar[1]{\oldmarginpar{\center{\footnotesize{\em #1}}}}

\newcommand{\comment}[1]{} 

\newcommand{\SH}{$\mathcal{SH}$\xspace}
\newcommand{\SHI}{$\mathcal{SHI}$\xspace}
\newcommand{\SROIQ}{$\mathcal{SROIQ}$\xspace}

\newcommand{\ALC}{$\mathcal{ALC}$\xspace} 
 
\newcommand{\SHIO}{$\mathcal{SHIO}$\xspace} 
\newcommand{\SHIQ}{$\mathcal{SHIQ}$\xspace} 
\newcommand{\SHOQ}{$\mathcal{SHOQ}$\xspace} 
\newcommand{\SHOQwxsp}{$\mathcal{SHOQ}$} 
\newcommand{\SHOIQ}{$\mathcal{SHOIQ}$\xspace}

\newcommand{\CSHOQ}{$C_\mathcal{SHOQ}$\xspace}

\newcommand{\CN}{\mathbf{C}} 
 
\newcommand{\RN}{\mathbf{R}} 
\newcommand{\IN}{\mathbf{I}} 

\newcommand{\closure}{\mathsf{closure}}

\newcommand{\Type}{\mathit{Type}}
\newcommand{\SType}{\mathit{SType}}
\newcommand{\Status}{\mathit{Status}}
\newcommand{\Label}{\mathit{Label}}

\newcommand{\EdgeLabels}{\mathit{EdgeLabels}}
\newcommand{\ELabels}{\mathit{ELabels}}
\newcommand{\piT}{\pi_T}
\newcommand{\piR}{\pi_R}
\newcommand{\piI}{\pi_I}

\newcommand{\RFormulas}{\mathit{RFmls}}

\newcommand{\DSTimeStamp}{\mathit{DSTimeStamp}}
\newcommand{\CSTimeStamp}{\mathit{CSTimeStamp}}
\newcommand{\ETimeStamp}{\mathit{ETimeStamp}}

\newcommand{\ILConstraints}{\mathit{ILConstraints}}

\newcommand{\State}{\mathsf{state}}
\newcommand{\NonState}{\mathsf{non\textrm{-}state}}
\newcommand{\Complex}{\mathsf{complex}}
\newcommand{\Simple}{\mathsf{simple}}

\newcommand{\True}{\mathsf{true}}
\newcommand{\False}{\mathsf{false}}
\newcommand{\Sat}{\mathsf{open}}
\newcommand{\Unsat}{\mathsf{closed}}
\newcommand{\Blocked}{\mathsf{blocked}}
\newcommand{\UnsatWrt}{\mathsf{closed\textrm{-}wrt}}

\newcommand{\PExpanded}{\mathsf{p\textrm{-}expanded}}
\newcommand{\FExpanded}{\mathsf{f\textrm{-}expanded}}
\newcommand{\Unexpanded}{\mathsf{unexpanded}}
\newcommand{\Null}{\mathsf{null}}

\newcommand{\TUS}{\mathsf{testingClosedness}}
\newcommand{\CQF}{\mathsf{checkingFeasibility}}

\SetKwInput{Input}{Input}
\SetKwInput{Output}{Output}
\SetKwInput{GlobalData}{Global data}
\SetKwInput{Purpose}{Purpose}
\SetKw{Call}{call}
\SetKw{Set}{set}
\SetKw{Exit}{exit}
\SetKw{Break}{break}
\SetKw{Continue}{continue}

\SetKwFunction{NewSucc}{NewSucc} 
\SetKwFunction{FindProxy}{FindProxy} 
\SetKwFunction{ConToSucc}{ConToSucc} 
\SetKwFunction{SRTR}{SRTR}
\SetKwFunction{Apply}{Apply}
\SetKwFunction{ApplyTransRule}{ApplyTransRule}
\SetKwFunction{ApplyConvRule}{ApplyConvRule}
\SetKwFunction{Tableau}{Tableau}
\SetKwFunction{Expand}{Expand}
\SetKwFunction{UpdateStatus}{UpdateStatus}
\SetKwFunction{UpdateUnsatNodes}{UpdateUnsatNodes}
\SetKwFunction{PropagateStatus}{PropagateStatus}
\SetKwFunction{Kind}{Kind}
\SetKwFunction{AfterTrans}{AfterTrans}
\SetKwFunction{FullLabel}{FullLabel}
\SetKwFunction{CEFullLabelR}{CEFullLabelR}
\SetKwFunction{TUnsat}{TUnsat}
\SetKwFunction{TSat}{TSat}
\SetKwFunction{ToExpand}{ToExpand}
\SetKwFunction{AFormulas}{AFmls}

\SetKwFunction{NDFormulas}{NDFmls}
\SetKwFunction{Cnj}{Cnj}

\newcommand{\Ext}{\mathit{ext}}
\newcommand{\Trans}[1]{\mathit{trans}_\mR(#1)}

\newcommand{\IFDL}[1]{\textrm{IFDL}(#1)}

\newcommand{\Size}{\mathsf{N}}

\newcommand{\UPS}{(\textrm{UPS})\xspace}
\newcommand{\UPSa}{(\textrm{UPS$_1$})\xspace}
\newcommand{\UPSb}{(\textrm{UPS$_2$})\xspace}
\newcommand{\UPSc}{(\textrm{UPS$_3$})\xspace}
\newcommand{\US}{(\textrm{US})\xspace}
\newcommand{\USa}{(\textrm{US$_1$})\xspace}
\newcommand{\USb}{(\textrm{US$_2$})\xspace}
\newcommand{\USc}{(\textrm{US$_3$})\xspace}

\newcommand{\FS}{(\textrm{FS})\xspace}

\newcommand{\DN}{(\textrm{DN})\xspace}

\newcommand{\NUS}{(\textrm{NUS})\xspace}
\newcommand{\TP}{(\textrm{TP})\xspace}
\newcommand{\TF}{(\textrm{TF})\xspace}

\newcommand{\Repl}{\mathit{IndRepl}}

\SetKwFunction{SetClosedWrt}{SetClosedWrt} 
\newcommand{\ILCh}{\mathit{ILC}}

\newtheorem{Explanation}{Explanation}
\newenvironment{explanation}{\begin{Explanation}\begin{em}}{\end{em}\end{Explanation}}

\newcommand{\sssection}[1]{\subsubsection{#1}}

\newcommand{\hasChild}{\mathit{hasChild}}
\newcommand{\hasParent}{\mathit{hasParent}}
\newcommand{\hasDescendant}{\mathit{hasDescendant}}
\newcommand{\Human}{\mathit{Human}}
\newcommand{\Male}{\mathit{Male}}
\newcommand{\Female}{\mathit{Female}}
\newcommand{\John}{\mathit{John}}
\newcommand{\Mary}{\mathit{Mary}}
\newcommand{\Doctor}{\mathit{Doctor}}
\newcommand{\Lawyer}{\mathit{Lawyer}}
\newcommand{\Teacher}{\mathit{Teacher}}
\newcommand{\Link}{\mathit{link}}
\newcommand{\Path}{\mathit{path}}
\newcommand{\Happy}{\mathit{Happy}}


%% file: SHOQ-main.tex
\section{Introduction}

Description logics (DLs) are formal languages suitable for representing terminological knowledge. They are of particular importance in providing a logical formalism for ontologies and the Semantic Web. 
DLs represent the domain of interest in terms of concepts, individuals, and roles. A concept is interpreted as a set of individuals, while a role is interpreted as a binary relation among individuals. A knowledge base in a DL consists of axioms about roles (grouped into an RBox), terminology axioms (grouped into a TBox), and assertions about individuals (grouped into an ABox). A DL is usually specified by: i) a set of constructors that allow building complex concepts and complex roles from concept names, role names and individual names, ii) allowed forms of axioms and assertions. 
The basic DL \ALC allows basic concept constructors listed in Table~\ref{table: constr-features}, but does not allow role constructors nor role axioms. The most common additional features for extending \ALC are also listed in Table~\ref{table: constr-features} together with syntax and examples: $\mathcal{I}$ is a role constructor, $\mathcal{Q}$ and $\mathcal{O}$ are concept constructors, while $\mathcal{H}$ and $\mathcal{S}$ are allowed forms of role axioms. 
The name of a DL is usually formed by the names of its additional features, as in the cases of \SH, \SHI, \SHIQ, \SHIO, \SHOQ and \SHOIQ. \SROIQ~\cite{HorrocksKS06} is a further expressive DL used as the logical base for the Web Ontology Language OWL~2~DL. 

\begin{table}[t]
\begin{center}
\begin{tabular}{|l|c|c|}
\hline
\multicolumn{3}{|c|}{Concept constructors of \ALC} \\ \hline
Constructor & Syntax & Example \\ \hline
complement & $\lnot C$ & $\lnot \Male$ \\ \hline
intersection & $C \mand D$ & $\Human \mand \Male$ \\ \hline
union & $C \mor D$ & $\Doctor \mor \Lawyer$ \\ \hline 
existential restriction & $\E r.C$ & $\E \hasChild.\Male$ \\ \hline
universal restriction & $\V r.C$ & $\V \hasChild.\Female$ \\ \hline
\hline
\multicolumn{3}{|c|}{Additional constructors/features of other DLs} \\ \hline
Constructor/Feature & Syntax & Example \\ \hline
\underline{i}nverse roles ($\mathcal{I}$) & $r^-$ & $\hasChild^-$ (i.e., $\hasParent$) \\ \hline
\underline{q}uantified number & $\geq\!n\,R.C$ & $\geq\!3\,\hasChild.\Male$ \\ 
restrictions ($\mathcal{Q}$) & $\leq\!n\,R.C$ & $\leq\!2\,\hasParent.\top$ \\ \hline
n\underline{o}minals ($\mathcal{O}$) & $\{a\}$ & $\{\John\}$ \\ \hline
\underline{h}ierarchies of roles ($\mathcal{H}$) & $R \sqsubseteq S$ & $\hasChild \sqsubseteq \hasDescendant$ \\ \hline
tran\underline{s}itive roles ($\mathcal{S}$) & $R \circ R \sqsubseteq R$ & $\hasDescendant \circ \hasDescendant \sqsubseteq \hasDescendant$ \\ \hline
\end{tabular}
\end{center}
\caption{Concept constructors for \ALC and some additional constructors/features of other DLs.\label{table: constr-features}}
\end{table}

Automated reasoning in DLs is useful, for example, in engineering and querying ontologies. One of basic reasoning problems in DLs is to check satisfiability of a knowledge base in a considered DL. Most of other reasoning problems in DLs are reducible to this one.
In this paper, we study the problem of checking satisfiability of a knowledge base in the DL \SHOQ, which extends the basic DL \ALC with transitive roles~($\mathcal{S}$), hierarchies of roles~($\mathcal{H}$), nominals~($\mathcal{O}$) and quantified number restrictions~($\mathcal{Q}$). It is known that this problem in \SHOQ is \EXPTIME-complete~\cite{DLnavigator} (even when numbers are coded in binary). 
Nominals, interpreted as singleton sets, are a useful notion to express identity and uniqueness. However, when interacting with inverse roles~($\mathcal{I}$) and quantified number restrictions in the DL~\SHOIQ, they cause the complexity of the above mentioned problem to jump up to \NEXPTIME-complete~\cite{TobiesThesis} (while that problem in any of the DLs \SHOQ, \SHIO, \SHIQ is \EXPTIME-complete~\cite{DLnavigator,HladikM04,TobiesThesis}). 

In~\cite{HorrocksS01} Horrocks and Sattler gave a tableau algorithm for deciding the DL \SHOQwxsp(D), which is the extension of \SHOQ with concrete datatypes. Later, Pan and Horrocks~\cite{PanH02} extended the method of~\cite{HorrocksS01} to give a tableau algorithm for deciding the DL \SHOQwxsp(D$_n$), which is the extension of \SHOQ with $n$-ary datatype predicates. 
These algorithms use backtracking to deal with disjunction ($\mor$) and ``or''-branching (e.g., the ``choice''-rule) and use a straightforward way for dealing with quantified number restrictions. They have a non-optimal complexity (\NtEXPTIME) when numbers are coded in unary.\footnote{When the algorithms are improved by using ``anywhere blocking'', the complexity will be \NEXPTIME and still non-optimal.}
In~\cite{FaddoulH10} Faddoul and Haarslev gave an algebraic tableau reasoning algorithm for \SHOQ, which combines the tableau method with linear integer programming. The aim was to increase efficiency of handling quantified number restrictions. However, their algorithm still uses backtracking to deal with disjunction and ``or''-branching and has a non-optimal complexity (``double exponential''~\cite{FaddoulH10}). 

This paper is a revised and extended version of our workshop paper~\cite{SHOQ-CSP}. In this work we present the {\em first} tableau method with an \EXPTIME (optimal) complexity for checking satisfiability of a knowledge base in the DL \SHOQ when numbers are coded in unary.\footnote{This corrects the claim of~\cite{SHOQ-CSP} that the complexity is measured using binary representation for numbers.} Our method is based on global caching and integer linear feasibility checking. 

The idea of global caching comes from Pratt's work~\cite{Pratt80} on PDL. It was formally formulated for tableaux in some DLs in~\cite{GoreNguyenTab07,GoreN11} and has been applied to several modal and description logics\LongVersion{~\cite{GoreNguyen05tab,GoreNguyen07clima,NguyenSzalas09ICCCI,NguyenSzalas-KSE09,NguyenS10FI,NguyenS10TCCI,NguyenS11SL,dkns2011}}\ShortVersion{ (see~\cite{SHOQ-long} for references)} to obtain tableau decision procedures with an optimal complexity. A variant of global caching, called global state caching, was used to obtain cut-free optimal tableau decision procedures for several modal logics with converse and DLs with inverse roles\LongVersion{~\cite{GoreW09,GoreW10,Nguyen-ALCI,SHI-ICCCI,SHIQ}}\ShortVersion{~\cite{GoreW09,GoreW10,SHI-ICCCI,SHIQ}}. 

Integer linear programming was exploited for tableaux in~\cite{Farsiniamarj08,FaddoulH10} to increase efficiency of reasoning with quantified number restrictions. However, the first work that applied integer linear feasibility checking to tableaux was~\cite{SHIQ-long,SHIQ}. In~\cite{SHIQ-long}, Nguyen gave the first \EXPTIME (optimal) tableau decision procedure for checking satisfiability of a knowledge base in the DL \SHIQ when numbers are coded in unary. His procedure is based on global state caching and integer linear feasibility checking. In the current paper, we apply his method of integer linear feasibility checking to \SHOQ. 
The adaptation requires special techniques due to the following reasons: i) we use global caching for \SHOQ, while Nguyen's work~\cite{SHIQ-long} uses global state caching for \SHIQ (for dealing with inverse roles); ii) we have to deal with the interaction between number restrictions and nominals.
Our method substantially differs from Farsiniamarj's method of exploiting integer programming for tableaux~\cite{Farsiniamarj08}. Our technique for dealing with both nominals and quantified number restrictions is also essentially different from the one by Faddoul and Haarslev~\cite{FaddoulH10}. 

\ShortVersion{Due to the lack of space, we restrict ourselves to introducing the problem of checking satisfiability of a knowledge base in the DL \SHOQ, describing our techniques for dealing with nominals and presenting some examples to illustrate our tableau method. In the long version~\cite{SHOQ-long} of this paper, we give a full description of our \EXPTIME tableau decision procedure for \SHOQ together with proofs of the results.} 
\LongVersion{The rest of this paper is structured as follows. In Section~\ref{section: prel} we recall notation and semantics of \SHOQ as well as the integer feasibility problem for DLs~\cite{SHIQ-long}. In Section~\ref{section: calculus} we present our tableau decision procedure for \SHOQ together with examples for illustrating our tableau method. We conclude this work in Section~\ref{section: conc}. Proofs for our results are given in the Appendix.}


\LongVersion{
\section{Preliminaries}\label{section: prel}

\subsection{Notation and Semantics of \SHOQ}
}

\ShortVersion{\section{Notation and Semantics of \SHOQ}\label{section: prel}}

Our language uses a finite set $\CN$ of {\em concept names}, a finite set $\RN$ of {\em role names}, and a finite set $\IN$ of {\em individual names}. 
%
%
We use letters like $A$ and $B$ for concept names, $r$ and $s$ for role names, and  $a$ and $b$ for individual names. We also refer to $A$ and $B$ as {\em atomic concepts}, to $r$ and $s$ as {\em roles}, and to $a$ and $b$ as {\em (named) individuals}. 

An (\SHOQ) {\em RBox} $\mR$ is a finite set of role axioms of the form $r \sqsubseteq s$ or $r \circ r \sqsubseteq r$. 
For example, $\Link \sqsubseteq \Path$ and $\Path \circ \Path \sqsubseteq \Path$ are such role axioms.

By $\Ext(\mR)$ we denote the least extension of $\mR$ such that:
\begin{itemize}
\item $r \sqsubseteq r \in \Ext(\mR)$ for any role $r$
\item if $r \sqsubseteq r' \in \Ext(\mR)$ and $r' \sqsubseteq r'' \in \Ext(\mR)$ then $r \sqsubseteq r'' \in \Ext(\mR)$.
\end{itemize}

We write $r \sqsubseteq_\mR s$ to denote $r \sqsubseteq s \in \Ext(\mR)$, and $\Trans{r}$ to denote $(r \circ r \sqsubseteq r) \in \Ext(\mR)$. If $r \sqsubseteq_\mR s$ then $r$ is a~{\em subrole} of $s$ (w.r.t.~$\mR$).
If $\Trans{s}$ then $s$ is a~{\em transitive role} (w.r.t.~$\mR$). 
A~role is {\em simple} (w.r.t.~$\mR$) if it is neither transitive nor has any transitive subrole (w.r.t.~$\mR$).

{\em Concepts} in \SHOQ\ are formed using the following BNF grammar, where $n$ is a nonnegative integer and $s$ is a simple role:
\[
C, D ::=
        \top
        \mid \bot
        \mid A
        \mid \lnot C
        \mid C \mand D
        \mid C \mor D
        \mid \E r.C
        \mid \V r.C
	\mid \{a\}
	\mid\ \geq\!n\,s.C
	\mid\ \leq\!n\,s.C
\]

A concept stands for a set of individuals. The concept $\top$ stands for the set of all individuals (in the considered domain). The concept $\bot$ stands for the empty set. The constructors $\lnot$, $\mand$ and $\mor$ stand for the set operators: complement, intersection and union. For the remaining forms, we just give some illustrative examples: $\E \hasChild.\Male$,\ \ $\V \hasChild.\Female$,\ \ $\geq\!2\,\hasChild.\Teacher$,\ \ $\leq\!5\,\hasChild.\top$.    

We use letters like $C$ and $D$ to denote arbitrary concepts.

A {\em TBox} is a~finite set of axioms of the form $C \sqsubseteq D$ or $C \doteq D$. 

An {\em ABox} is a~finite set of assertions of the form $a\!:\!C$, $r(a,b)$ or $a \not\doteq b$.
An {\em eABox} ({\em extended ABox}) is a~finite set of assertions of the form $a\!:\!C$, $r(a,b)$, $\lnot r(a,b)$, $a \doteq b$ or $a \not\doteq b$.

An axiom $C \sqsubseteq D$ means $C$ is a subconcept of $D$, while $C \doteq D$ means $C$ and $D$ are equivalent concepts. 
An assertion $a\!:\!C$ means $a$ is an instance of concept $C$, 
$r(a,b)$ means the pair $\tuple{a,b}$ is an instance of role $r$, 
and $a \not\doteq b$ means $a$ and $b$ are distinct individuals.

A {\em knowledge base} in \SHOQ\ is a tuple $\tuple{\mR,\mT,\mA}$, where $\mR$ is an RBox, $\mT$ is a~TBox and $\mA$ is an ABox.

We say that a role $s$ is {\em numeric} w.r.t. a knowledge base $\KB = \tuple{\mR,\mT,\mA}$ if: 
\begin{itemize}
\item it is simple w.r.t.~$\mR$ and occurs in a concept $\geq\!n\,s.C$ or $\leq\!n\,s.C$ in $\KB$, or 
\item $s \sqsubseteq_\mR r$ and $r$ is numeric w.r.t.~$\KB$. 
\end{itemize}
We will simply call such an $s$ a numeric role when $\KB$ is clear from the context. 

A {\em formula} is defined to be either a concept or an eABox assertion. 
We use letters like $\varphi$, $\psi$ and $\xi$ to denote formulas.
%
Let $\Null\!:\!C$ stand for $C$. We use $\alpha$ to denote either an individual or $\Null$. Thus, $\alpha\!:\!C$ is a formula of the form $a\!:\!C$ or $\Null\!:\!C$ (which means $C$). 

An {\em interpretation} $\mI = \langle \Delta^\mI, \cdot^\mI
\rangle$ consists of a~non-empty set $\Delta^\mI$, called the {\em
domain} of $\mI$, and a~function $\cdot^\mI$, called the {\em
interpretation function} of $\mI$, that maps each concept name
$A$ to a~subset $A^\mI$ of $\Delta^\mI$, each role name
$r$ to a~binary relation $r^\mI$ on $\Delta^\mI$, and each individual name $a$ to an element $a^\mI \in \Delta^\mI$.
The interpretation function is extended to complex concepts as
follows, where $\sharp Z$ denotes the cardinality of a set~$Z$: 
\[
\begin{array}{c}
\top^\mI = \Delta^\mI \quad\quad
\bot^\mI = \emptyset \quad\quad
(\lnot C)^\mI = \Delta^\mI - C^\mI
\\[1.0ex]
(C \mand D)^\mI = C^\mI \cap D^\mI \quad\quad
(C \mor D)^\mI = C^\mI \cup D^\mI \quad\quad
\{a\}^\mI = \{a^\mI\}
\\[1.0ex]
(\E r.C)^\mI =
        \big\{ x \in \Delta^\mI \mid \E y\big[ \tuple{x,y} \in r^\mI
                 \textrm{ and } y \in C^\mI\big]\big\}
\\[1.0ex]
(\V r.C)^\mI =
        \big\{ x \in \Delta^\mI \mid \V y\big[ \tuple{x,y} \in r^\mI
                       \textrm{ implies } y \in C^\mI\big]\big\}
\\[1.0ex]
(\geq\!n\,s.C)^\mI =
        \big\{ x \in \Delta^\mI \mid \sharp\{y \mid \tuple{x,y} \in s^\mI
                 \textrm{ and } y \in C^\mI \} \geq n \big\}
\\[1.0ex]
(\leq\!n\,s.C)^\mI =
        \big\{ x \in \Delta^\mI \mid \sharp\{y \mid \tuple{x,y} \in s^\mI
                 \textrm{ and } y \in C^\mI \} \leq n \big\}.
\end{array}
\]

For a set $\Gamma$ of concepts, define $\Gamma^\mI = \{x \in \Delta^\mI \mid x \in C^\mI \textrm{ for all } C \in \Gamma\}$.

The relational composition of binary relations $R_1$ and $R_2$ is denoted by $R_1 \circ R_2$.

An interpretation $\mI$ is a~{\em model of an RBox} $\mR$ if for every axiom $r \sqsubseteq s$ (resp.\ $r \circ r \sqsubseteq r$) of $\mR$, we have that $r^\mI \subseteq s^\mI$ (resp.\ $r^\mI \circ r^\mI \subseteq r^\mI$). 
Note that if $\mI$ is a model of $\mR$ then it is also a model of $\Ext(\mR)$.

An interpretation $\mI$ is a~{\em model of a~TBox} $\mT$ if for
every axiom $C \sqsubseteq D$ (resp.\ $C \doteq D$) of $\mT$,
we have that $C^\mI \subseteq D^\mI$ (resp.\ $C^\mI = D^\mI$).

Given an interpretation $\mI$, define:
\[
\begin{array}{lcl}
\mI \models a\!:\!C & \;\textrm{iff}\; & a^\mI \in C^\mI \\
\mI \models r(a,b) & \;\textrm{iff}\; & \tuple{a^\mI,b^\mI} \in r^\mI \\
\mI \models \lnot r(a,b) & \;\textrm{iff}\; & \tuple{a^\mI,b^\mI} \notin r^\mI \\
\mI \models a \doteq b & \;\textrm{iff}\; & a^\mI = b^\mI \\
\mI \models a \not\doteq b & \;\textrm{iff}\; & a^\mI \neq b^\mI.
\end{array}
\]
If $\mI \models \varphi$ then we say that $\mI$ {\em satisfies} $\varphi$. 
An interpretation $\mI$ is a~{\em model of an eABox} $\mA$ if it satisfies all the assertions of $\mA$. In that case, we also say that $\mI$ {\em satisfies} $\mA$. 

An interpretation $\mI$ is a~{\em model of a~knowledge base}
$\tuple{\mR,\mT,\mA}$ if $\mI$ is a~model of $\mR$, $\mT$ and~$\mA$.
A~knowledge base $\tuple{\mR,\mT,\mA}$ is {\em satisfiable} if it has
a~model.

An interpretation $\mI$ {\em satisfies} a concept $C$ (resp.\ a set $X$ of concepts) if $C^\mI \neq \emptyset$ (resp.\ $X^\mI \neq \emptyset$).
It {\em validates} a concept $C$ if $C^\mI = \Delta^\mI$. 
A set $X$ of concepts is {\em satisfiable w.r.t.\ an RBox $\mR$ and a TBox $\mT$} if there exists a model of $\mR$ and $\mT$ that satisfies $X$. 
We say that an eABox $\mA$ is {\em satisfiable w.r.t.\ an RBox $\mR$ and a TBox $\mT$} if there exists an interpretation $\mI$ that is a model of $\mA$, $\mR$ and $\mT$. In that case, we also call $\mI$ a model of $\tuple{\mR,\mT,\mA}$. 

In this paper, we assume that concepts and ABox assertions are represented in
negation normal form (NNF), where $\lnot$ occurs only directly before atomic concepts.\footnote{Every formula can be transformed to an equivalent formula in NNF in polynomial time.} 
We use $\overline{C}$ to denote the NNF of $\lnot C$, and for $\varphi = (a\!:\!C)$, we use $\ovl{\varphi}$ to denote $a\!:\!\ovl{C}$.
For simplicity, we treat axioms of a TBox $\mT$ as concepts representing global assumptions: an axiom $C \sqsubseteq D$ is treated as $\overline{C} \mor D$, while an axiom $C \doteq D$ is treated as $(\overline{C} \mor D) \mand (\overline{D} \mor C)$.\footnote{As this way of handling the TBox is not efficient in practice, the absorption technique like the one discussed in~\cite{SHI-ICCCI} can be used to improve the performance of reasoning.} 
That is, we assume that $\mT$ consists of concepts in NNF. A~concept $C \in \mT$ can be thought of as an axiom $\top \sqsubseteq C$. Thus, an interpretation $\mI$ is a~model of $\mT$ iff $\mI$ validates every concept $C \in \mT$. 


\LongVersion{
\subsection{An Integer Feasibility Problem for Description Logics}

For dealing with number restrictions in \SHOQ, we consider the following integer feasibility problem, which was introduced in~\cite{SHIQ-long}:
\[
\begin{array}{c}
\displaystyle \sum_{j=1}^m a_{i,j}\cdot x_j \,\bowtie_i\, b_i,\  
	\textrm{ for } 1 \leq i \leq l;\\[1.5ex]
x_j \geq 0,\ 
	\textrm{ for } 1 \leq j \leq m;
\end{array}
\]
where each $a_{i,j}$ is either 0 or 1, each $x_j$ is a variable standing for a natural number, each $\bowtie_i$ is either $\leq$ or $\geq$, each $b_i$ is a natural number encoded by using no more than $n$ bits (i.e., $b_i \leq 2^n$). We call this an $\IFDL{l,m,n}$-problem (a~problem of Linear Integer Feasibility for Description Logics with size specified by $l,m,n$). The problem is {\em feasible} if it has a solution (i.e., values for the variables $x_j$, $1 \leq j \leq l$, that are natural numbers satisfying the constraints), and is {\em infeasible} otherwise. By solving an $\IFDL{l,m,n}$-problem we mean checking its feasibility. 

It is known from linear programming that, if the variables $x_j$ are not required to be natural numbers but can be real numbers then the above feasibility problem can be solved in polynomial time in $l$, $m$ and $n$. The general integer linear optimization problem is known to be NP-hard.\footnote{\url{http://en.wikipedia.org/wiki/Integer_programming}}

To solve an integer feasibility problem, we propose to use the decomposition technique and the ``branch and bound'' method~\cite{BranchAndBound}. One can first analyze dependencies between the variables and the constraints to decompose the problem into smaller independent subproblems, then solve the subproblems that are trivial, and after that apply the ``branch and bound'' method~\cite{BranchAndBound} to the remaining subproblems. 

The above mentioned approach may not guarantee that a given $\IFDL{l,m,n}$-problem is solved in exponential time in~$n$. We recall below an estimation of the upper bound for the complexity for some specific cases, using another approach. 

\newcommand{\LemmaIFDL}{Every $\IFDL{l,m,n}$-problem such that 
$l \leq n$, 
$m$ is (at most) exponential in $n$, 
and $b_i \leq n$ for all $1 \leq i \leq l$ 
can be solved in (at most) exponential time in~$n$.}
\begin{lemma}[\cite{SHIQ-long}]\label{lemma: IFDL}
\LemmaIFDL
\end{lemma}

\begin{proof}
Consider the following nondeterministic procedure:
\begin{enumerate}
\item initialize $c_{i,j} := 0$ for each $1 \leq i \leq l$ and $1 \leq j \leq m$ such that $a_{i,j} = 1$
\item for each $i$ from 1 to $l$ do
  \begin{itemize}
  \item[] for each $k$ from 1 to $b_i$ do
     \begin{itemize}
     \item[] choose some $j$ among $1, \ldots, m$ such that $a_{i,j} = 1$ and set $c_{i,j} := c_{i,j} + 1$
     \end{itemize}
  \end{itemize}
\item if the set of constraints $\{x_j \,\bowtie_i\, c_{i,j} \mid 1 \leq i \leq l, 1 \leq j \leq m, a_{i,j} = 1\}$ is feasible then return ``yes'', else return ``no''.
\end{enumerate}

Observe that the considered $\IFDL{l,m,n}$-problem is feasible iff there exists a run of the above procedure that returns ``yes''. Since $b_i \leq n$ for all $1 \leq i \leq l$, there are no more than $m^{l\cdot n}$ possible runs of the above procedure. All the steps of the procedure can be executed in time $O(l\cdot m \cdot n)$. Since $l \leq n$ and $m$ is (at most) exponential in $n$, we conclude that the considered $\IFDL{l,m,n}$-problem can deterministically be solved in (at most) exponential time in~$n$.
\myEnd
\end{proof}

The following lemma is more general than the above lemma. 

\begin{lemma}[\cite{SHIQ-long}]\label{lemma: IFDL2}
Every $\IFDL{l,m,n}$-problem satisfying the following properties can be solved in (at most) exponential time in~$n\,$: 
\begin{itemize}
\item $l \leq n$, $m$ is (at most) exponential in $n$, 
\item and 
   \begin{itemize}
   \item either $b_i \leq n$ for all $1 \leq i \leq l$ such that $\bowtie_i$ is $\leq$
   \item or $b_i \leq n$ for all $1 \leq i \leq l$ such that $\bowtie_i$ is $\geq$.
   \end{itemize}
\end{itemize}
\end{lemma}

\begin{proof}
Suppose $l \leq n$, $m$ is (at most) exponential in $n$, and $b_i \leq n$ for all $1 \leq i \leq l$ such that $\bowtie_i$ is $\leq$. The other case is similar and omitted. Consider the following nondeterministic procedure:
\begin{enumerate}
\item let $J = \{ j \mid 1 \leq j \leq m$ and there exists $1 \leq i \leq l$ such that $\bowtie_i$ is $\leq$ and $a_{i,j} = 1\}$
\item for each $1 \leq i \leq l$ and $1 \leq j \leq m$ such that $\bowtie_i$ is $\leq$ and $a_{i,j} = 1$, set $c_{i,j} := 0$
\item for each $i$ from 1 to $l$ such that $\bowtie_i$ is $\leq$, do
  \begin{itemize}
  \item[] for each $k$ from 1 to $b_i$ do
     \begin{itemize}
     \item[] choose some $j$ among $1, \ldots, m$ such that $a_{i,j} = 1$ and set $c_{i,j} := c_{i,j} + 1$
     \end{itemize}
  \end{itemize}
\item for each $j \in J$ do
  \begin{itemize}
  \item[] $d_j := \min\{c_{i,j} \mid 1 \leq i \leq l,\; \bowtie_i$ is $\leq$ and $a_{i,j} = 1\}$
  \end{itemize}
\item if the set of constraints $\{\sum_{j=1}^m a_{i,j}\cdot x_j \geq b_i \mid 1 \leq i \leq l,\; \bowtie_i$ is $\geq\} \cup \{x_j = d_j \mid j \in J\}$ is feasible then return ``yes'', else return ``no''.
\end{enumerate}

Observe that the considered $\IFDL{l,m,n}$-problem is feasible iff there exists a run of the above procedure that returns ``yes''. Under the assumptions of the lemma, there are no more than $m^{l\cdot n}$ possible runs of the above procedure. All the steps of the procedure can be executed in time $O(l\cdot m \cdot n)$. Since $l \leq n$ and $m$ is (at most) exponential in $n$, we conclude that the considered $\IFDL{l,m,n}$-problem can deterministically be solved in (at most) exponential time in~$n$.
\myEnd
\end{proof}
}


\LongVersion{
\section{The Traditional Tableau Method and Its Problems}
\label{section: HDSJM}

The problem we study is to check whether a given knowledge base $\KB = \tuple{\mR,\mT,\mA}$ in \SHOQ is satisfiable. 
The {\em traditional tableau method} for this task is as follows. We start from the ABox $\mA$ and try to modify it to obtain a model of $\KB$. At each moment, we have an ABox, which is like a graph. 
At the beginning, each (named) individual occurring in $\mA$ is a node labeled by the set $\Label(a) = \{C \mid a\!:\!C \in \mA\} \cup \mT$, and each assertion $r(a,b)$ in $\mA$ forms an edge from~$a$ to~$b$ that is labeled by~$r$. 
The concepts in $\Label(a)$ are treated as requirements to be realized for $a$. As $\mT$ consists of the global assumptions that should be satisfied for all individuals, the concepts from $\mT$ are included in $\Label(a)$. For example, an axiom $\top \sqsubseteq \Human$ is encoded in NNF as $\Human$, and such a global assumption states that all individuals in the domain should be human beings. 
To see how the requirements for nodes can be realized, let us consider several cases:
\begin{itemize}
\item If $C \mand D \in \Label(v)$ then to realize the requirement $C \mand D$ for $v$ we add both $C$ and $D$ to $\Label(v)$. To see the intuition of this, assume that $\John$ is an individual and $\Male \mand \Happy \in \Label(\John)$. In this case, $\John$ is required to satisfy the property $\Male \mand \Happy$, and to realize this, we add both the requirements $\Male$ and $\Happy$ to $\Label(\John)$. 

\item If $C \mor D \in \Label(v)$ then to realize the requirement $C \mor D$ for $v$ we add either $C$ or $D$ to $\Label(v)$. That is, we make an ``or''-branching, which is dealt with by backtracking (since at each moment we consider only one ABox). If the current ``or''-branch leads to inconsistency, we will backtrack to the nearest ``or''-branching point and try another ``or''-branch. To see the intuition of this, assume that $\Doctor \mor \Lawyer \in \Label(\John)$. In this case, $\John$ is required to satisfy the property $\Doctor \mor \Lawyer$, which states that he is either a doctor or a lawyer, and to realize this requirement, we make a choice: either add the requirement $\Doctor$ or add the requirement $\Lawyer$ to $\Label(\John)$. 

\item If $\E r.C \in \Label(v)$ then to realize the requirement $\E r.C$ for $v$ we connect $v$ to a new node $w$ with $\Label(w) = \{C\} \cup \mT$ via an edge labeled by $r$. Once again, $\mT$ is included in $\Label(w)$ because it consists of the global assumptions that should be realized for all individuals. (Instead of creating a new node, one may use an existing node for $w$ as in the approach with global caching, but this should be done appropriately, e.g., as in our tableau method discussed in the next section. Alternatively, one can use a blocking technique as in~\cite{HorrocksS01,PanH02}.) To see the intuition of the expansion, assume that $\E\hasChild.\Female \in \Label(\John)$. In this case, $\John$ should satisfy the requirement $\E\hasChild.\Female$, which states that he has a female child (a daughter). To realize this, we connect the node $\John$ of the graph to a new node $w$ with $\Label(w) = \{\Female\} \cup \mT$ via an edge labeled by $\hasChild$. From this, it can be seen that the graph contains not only named individuals occurring in $\mA$, but it may also contain nodes like $w$, which are called unnamed individuals. 

\item If $\V r.C \in \Label(v)$ then to realize the requirement $\V r.C$ for $v$, for every node $w$ such that there is an edge with the label $r$ from $v$ to $w$, we add $C$ to $\Label(w)$. To see the intuition of this, assume that $\V\hasChild.\Happy \in \Label(\John)$ and there are edges with the label $\hasChild$ from the node $\John$ to the nodes $\Mary$ and $w$ (i.e., $\Mary$ and $w$ are children of $\John$). In this case, $\John$ should satisfy the requirement $\V\hasChild.\Happy$, which states that all the children of $\John$ should be happy. To realize this, we add the requirement $\Happy$ to both $\Label(\Mary)$ and $\Label(w)$.

\item If $\{a\} \in \Label(v)$ then $v$ and $a$ should denote the same individual (this is the semantics of nominals), and to realize the requirement $\{a\}$ for $v$ we merge the nodes $v$ and $a$ together in an appropriate way. 

\item If $\geq\!n\,r.C \in \Label(v)$ then to realize the requirement $\geq\!n\,r.C$ for $v$ we connect $v$ to $n$ new nodes $w_1,\ldots,w_n$ with $\Label(w_i) = \{C\} \cup \mT$ via an edge labeled by $r$ for all $1 \leq i \leq n$, and keep the constraints $w_i \not\doteq w_j$ for all $1 \leq i \neq j \leq n$. (Once again, an appropriate caching or blocking technique can be used to reduce the number of created nodes.) To see the intuition of the expansion, assume that $\geq\!2\,\hasChild.\Female \in \Label(\John)$. In this case, $\John$ should satisfy the requirement $\geq\!2\,\hasChild.\Female$, which states that he has at least two female children (daughters). To realize this, we connect the node $\John$ of the graph to new nodes $w_1$ and $w_2$ with $\Label(w_1) = \Label(w_2) = \{\Female\} \cup \mT$ via an edges labeled by $\hasChild$ and keep the constraint $w_1 \not\doteq w_2$. 

\item If $\leq\!n\,r.C \in \Label(v)$ and there are pairwise different nodes $w_1,\ldots,w_{n+1}$ such that $v$ is connected to $w_i$ via an edge labeled by $r$ and $C \in \Label(w_i)$ for all $1 \leq i \leq n+1$, then:
  \begin{itemize}
  \item if there exist different $i$ and $j$ among $1,\ldots, n$ such that the constraint $w_i \not\doteq w_j$ is absent then we merge $w_i$ and $w_j$ together in an appropriate way,
  \item otherwise, the current ABox is inconsistent and we do backtracking.
  \end{itemize}
\end{itemize}

Inconsistency may occur, for example, in the following cases:
\begin{itemize}
\item when $\bot \in \Label(v)$ for some $v$; or
\item when $\{A,\lnot A\} \subseteq \Label(v)$ for some $A$ and $v$; or
\item when $a$ and $b$ were merged together, but the assertion $a \not\doteq b$ is a kept constraint; or 
\item when the current ABox contains an edge with the label $r$ from $a$ to $b$, but $(\lnot r(a,b)) \in \mA$.
\end{itemize}

As mentioned before, when the current ABox is inconsistent, backtracking occurs, and if there is no ``or''-branching point to come back, the process terminates with the result ``$\KB$ is unsatisfiable''. 

The above discussion only gives a sketch on how the traditional tableau method works. We did not discuss how role axioms can be dealt with and how a blocking technique can be applied to guarantee termination. Furthermore, merging nodes causes merging edges, and hence an edge may be labeled by a set of roles. 
In general, a tableau algorithm is usually designed so that, if it does not terminate with the result ``$\KB$ is unsatisfiable'', then $\KB$ is satisfiable and we can directly construct a model of $\KB$ from the resulting ({\em clash-free} and {\em completed}) ABox. We refer the reader to~\cite{HorrocksS01,PanH02} for details. 

The traditional tableau method for \SHOQ has the advantage of being intuitive, but it has two disadvantages that make the complexity non-optimal (\NtEXPTIME or \NEXPTIME, depending on the applied blocking technique, in comparison with the optimal complexity \EXPTIME) and the reasoning process not scalable w.r.t.\ number restrictions: 
\begin{itemize}
\item An ABox is like an ``and''-structure (i.e., all of its assertions must hold together) and the search space for the traditional tableau method is an ``or''-tree of ``and''-structures. Recall that ``or''-branchings are caused, amongst others, by the rule for realizing requirements of the form $C \mor D$. The problem is that two nodes in ABoxes in different ``or''-branches may have the same label and the same ``neighborhood'', and both of them are expanded with no reuse, which causes a kind of redundant computation~\cite{GoreN11}. 
\item Reconsider the traditional tableau rule for realizing a requirement of the form $\geq\!n\,r.C$. If $n$ is big, for example, 1000 or 1000000, then the rule creates many new nodes. In the DL literature, this is called ``pay-as-you-go'', but this payment is unnecessarily too high when $n$ is big and it causes the reasoning process not scalable w.r.t.\ number restrictions.
\end{itemize}
} 


\LongVersion{
\section{\EXPTIME Tableaux for \SHOQ}
\label{section: calculus}

In this section, we first define the data structures and outline the framework of our tableau method. We then describe our techniques for dealing with nominals. After that, we specify the used tableau rules and state properties of the resulting tableau decision procedure.
}

\ShortVersion{\section{An \EXPTIME Tableau Method for \SHOQ}

In this section, we first define the data structures and outline the framework of our tableau method. We then describe our techniques for dealing with nominals. After that, we state properties of the resulting tableau decision procedure.
} 


\subsection{Data Structures and the Tableau Framework}

\LongVersion{Recall that the search space for the traditional tableau method for \SHOQ~\cite{HorrocksS01,PanH02} is an ``or''-tree of ``and''-structures, and this causes the complexity of the reasoning process to become non-optimal (even in the case without number restrictions).}
\ShortVersion{As discussed in the long version~\cite{SHOQ-long} of the current paper, the search space of the traditional tableau method for \SHOQ~\cite{HorrocksS01,PanH02} is an ``or''-tree of ``and''-structures, and this causes the complexity of the reasoning process to become non-optimal (even in the case without number restrictions).} 
The idea for overcoming this problem is to use global caching~\cite{Pratt80,GoreN11,SHI-ICCCI}. With global caching, the search space is like a single ``and-or'' graph. For checking satisfiability of a concept w.r.t.\ an RBox and a TBox~\cite{GoreN11}, each node of the graph is a {\em simple node} like an individual (in an ABox). For checking satisfiability of a knowledge base~\cite{SHI-ICCCI,SHIQ-long}, each node of the graph is either a {\bf complex node} like an eABox, or a {\bf simple node} like an individual. More precisely, the label of a complex node is a set of eABox assertions, while the label of a simple node is a set of concepts. \ShortVersion{The {\bf label of a node $v$} is denoted by $\Label(v)$. The formulas in $\Label(v)$ are treated as requirements to be realized for the node~$v$.} The information about whether a node $v$ is complex or simple is kept by $\SType(v)$ (the {\bf subtype of $v$}). 

At the beginning, the graph has only one node, called the {\bf root}, which is a complex node. Then, in the first stage, complex nodes are expanded only by so called {\bf static (tableau) rules} that do not create new (unnamed) individuals. This creates a layer of complex nodes. When no static tableau rules are applicable to a complex node $v$, if $\Label(v)$ contains a requirement of the form $a\!:\!\E r.C$ then to realize this requirement we can connect $v$ to a {\em simple} node $w$ with $\Label(w) = \{C\} \cup \mT$ via an edge~$e$. This edge is related to $a$ and $r$. To keep this information we store $\piI(e) = a$ (the letter $\pi$ stands for ``projection'' and the letter $I$ stands for ``individual'') and $\piR(e) = \{r\}$ (the letter $R$ stands for ``roles''; as mentioned \LongVersion{earlier}\ShortVersion{in~\cite{SHOQ-long}}, due to merging nodes, an edge may be labeled by more than one role, and hence we use a set of roles).

A {\bf transitional (tableau) rule} is a rule that realizes a requirement of the form $a\!:\!\E r.C$, $a\!:(\geq\!n\,r.C)$, $\E r.C$ or $\geq\!n\,r.C$ for a node $v$ by connecting $v$ to a new node or a number of new nodes or by using some existing nodes. If no static rules are applicable to a node $v$ then $v$ is called a {\bf state}, otherwise it is called a {\bf non-state}. This information is kept by $\Type(v)$ (the {\bf type of $v$}). Transitional tableau rules are applied only to states. A non-state is like an ``or''-node in an ``and-or'' graph, but a state is a structure more sophisticated than an ``and''-node in an ``and-or'' graph (due to feasibility checking of the set of integer linear constraints related to the state).\footnote{In tableaux for simpler DLs like \ALC~\cite{GoreN11} or \SHI~\cite{SHI-ICCCI}, a state is simply an ``and''-node.} 

Consider a simple state $v$ (i.e., a state that is a simple node) with $\E r.C \in \Label(v)$. To realize this requirement for $v$, we can connect $v$ to a new simple node $w$ with $\Label(w) = \{C\} \cup \mT$ by an edge~$e$. For such an edge~$e$, let $\piI(e) = \Null$ (i.e., no named individual is related to~$e$). 

Consider a state $v$. To realize requirements of the form $a\!:\!(\geq\!n\,r.C)$, $a\!:\!(\leq\!n\,r.C)$, $\geq\!n\,r.C$ or $\leq\!n\,r.C$ for $v$, we may have to connect $v$ to some simple nodes $w_i$ by edges $e_i$, respectively, and check feasibility of a certain set of integer linear constraints. The {\bf set of integer linear constraints for $v$} is kept by $\ILConstraints(v)$. For such mentioned edges $e_i$, let $\piT(e_i) = \CQF$ (the letter $T$ stands for ``type''). For other edges $e$, which are created for realizing a requirement of the form $a\!:\!\E r.C$ or $\E r.C$, let $\piT(e) = \TUS$. 

We have explained the attributes $\piT(e)$, $\piR(e)$ and $\piI(e)$ that should be kept for an edge $e$ outgoing from a state. Summing up, we have the following formal definition:
\begin{definition}\LongVersion{\em}
Let $\EdgeLabels = \{\TUS, \CQF\} \times \mathcal{P}(\RN) \times (\IN \cup \{\Null\})$. For $e \in \EdgeLabels$, let $e = \tuple{\piT(e),\piR(e),\piI(e)}$. Thus, $\piT(e)$ is called the type of the edge label $e$, $\piR(e)$ is a set of roles, and $\piI(e)$ is either an individual or $\Null$. (Each edge is specified by the source, the target and the label.)
\myEnd
\end{definition}

We have explained the attributes $\Label(v)$, $\Type(v)$, $\SType(v)$ and $\ILConstraints(v)$ for a node $v$. We need three more attributes for $v$, which are described and justified below.
\begin{itemize}
\item To realize the requirement $C \mor D \in \Label(v)$ for a simple node $v$, we expand $v$ by a static rule that connects $v$ to two simple nodes $w_1$ and $w_2$ such that $\Label(w_1) = \Label(v) \cup \{C\} \setminus \{C \mor D\}$ and $\Label(w_2) = \Label(v) \cup \{D\} \setminus \{C \mor D\}$. The requirement $C \mor D$ is put to the sets $\RFormulas(w_1)$ and $\RFormulas(w_2)$ to record that it has been realized for $w_1$ and $w_2$, respectively. In general, the attribute $\RFormulas(w)$ for a node $w$ keeps the set of the requirements that have been realized for $w$ by using static rules. It is called the {\bf set of reduced formulas of~$w$}.  

\item Suppose $v$ is a complex node and either $a \doteq b$ or $a\!:\!\{b\}$ belongs to $\Label(v)$. Then, to realize that requirement for $v$, we merge the individual $b$ to the individual $a$ in an appropriate way and record this fact by keeping $\Repl(v)(b) = a$. The attribute $\Repl(v)$ is called the {\bf partial mapping specifying replacements of individuals for the node $v$}.

\item The last attribute needed for a node $v$ is called the {\bf status of $v$} and denoted by $\Status(v)$. Possible statuses of nodes are: unexpanded, partially-expanded, fully-expanded, closed, open, blocked, and closed w.r.t.\ a set of complex states. Informally, $\Unsat$ means ``unsatisfiable w.r.t.\ $\mR$ and $\mT$'', $\Sat$ means ``satisfiable w.r.t.\ $\mR$ and $\mT$'', and $\UnsatWrt(U)$ means ``unsatisfiable w.r.t.\ $\mR$, $\mT$ and any node from $U$''. 
\end{itemize}
 
We arrive at the following formal definition.

\begin{definition}\LongVersion{\em}
A {\em tableau} is a rooted graph $G = \tuple{V,E,\nu}$, where $V$ is a set of nodes, $E \subseteq V \times V$ is a set of edges, $\nu \in V$ is the root, each node $v \in V$ has a number of attributes, and each edge $\tuple{v,w}$ may have a number of labels from $\EdgeLabels$.\footnote{An edge $\tuple{v,w}$ may have a number of labels from $\EdgeLabels$ because of global caching, which we will briefly discuss later.} 
The attributes of a tableau node $v$ are:
\begin{itemize}
\item $\Type(v) \in \{\State, \NonState\}$. 

\item $\SType(v) \in \{\Complex, \Simple\}$ is called the subtype of $v$. 

\item $\Status(v) \in \{\Unexpanded$, $\PExpanded$, $\FExpanded$, $\Unsat$, $\Sat$, $\Blocked\} \cup \{\UnsatWrt(U) \mid$ $U \subseteq V$ and $\Type(u) = \State \land \SType(u) = \Complex$ for all $u \in U\}$, where $\PExpanded$ and $\FExpanded$ mean ``partially expanded'' and ``fully expanded'', respectively. $\Status(v)$ may be $\PExpanded$ only when $\Type(v) = \State$. 
If $\Status(v) = \UnsatWrt(U)$ then we say that the node $v$ is closed w.r.t.\ any node from $U$. 

\item $\Label(v)$ is a finite set of formulas called the label of $v$. 

\item $\RFormulas(v)$ is a finite set of formulas called the set of reduced formulas of~$v$.

\item $\Repl(v) : \IN \to \IN$ is a partial mapping specifying replacements of individuals. It is available only when $v$ is a complex node. If $\Repl(v)(a) = b$ then at the node $v$ we have $a \doteq b$ and $b$ is the representative of its equivalence class. 

\item $\ILConstraints(v)$ is a set of integer linear constraints. It is available only when $\Type(v) = \State$. The constraints use variables $x_{w,e}$ indexed by a pair $\tuple{w,e}$ such that $\tuple{v,w} \in E$, $e \in \ELabels(v,w)$ and $\piT(e) = \CQF$. Such a variable specifies how many copies of the successor $w$ using the edge label $e$ will be created for~$v$. 
\myEnd
\end{itemize}
\end{definition}

If $\tuple{v,w} \in E$ then we call $v$ a {\em predecessor} of $w$ and $w$ a {\em successor} of~$v$. 
An edge outgoing from a node $v$ has labels iff $\Type(v) = \State$. 
When defined, the set of labels of an edge $\tuple{v,w}$ is denoted by $\ELabels(v,w)$.   
If $e \in \ELabels(v,w)$ then $\piI(e) = \Null$ iff $\SType(v) = \Simple$.  

Formally, a node $v$ is called a {\em state} if $\Type(v) = \State$, and a {\em non-state} otherwise. It is called a {\em complex node} if $\SType(v) = \Complex$, and a {\em simple node} otherwise. 
The root $\nu$ is a complex non-state. 

A node may have status $\Blocked$ only when it is a simple node with the label containing a nominal $\{a\}$. The status $\Blocked$ can be updated only to $\Unsat$ or $\UnsatWrt(\ldots)$. We write $\UnsatWrt(\ldots)$ to mean $\UnsatWrt(U)$ for some~$U$. By $\Status(v) \neq \UnsatWrt(\{u,\ldots\})$ we denote that $\Status(v)$ is not of the form $\UnsatWrt(U)$ with $u \in U$. 

The graph $G$ consists of two layers: the layer of complex nodes and the layer of simple nodes. There are no edges from simple nodes to complex nodes. The edges from complex nodes to simple nodes are exactly the edges outgoing from complex states. That is, if $\tuple{v,w}$ is an edge from a complex node $v$ to a simple node $w$ then $\Type(v) = \State$, if $\Type(v) = \State$ and $\tuple{v,w} \in E$ then $\SType(w) = \Simple$. 
Each complex node of $G$ is like an eABox (more formally, its label is an eABox), which can be treated as a graph whose vertices are named individuals. On the other hand, a simple node of $G$ stands for an unnamed individual. 
If $e$ is a label of an edge from a complex state $v$ to a simple node $w$ then the triple $\tuple{v,e,w}$ can be treated as an edge from the named individual $\piI(e)$ (an inner node in the graph representing~$v$) to the unnamed individual corresponding to $w$, and that edge is via the roles from~$\piR(e)$.

We will use also assertions of the form $a\!:\!(\preceq\!n\,s.C)$ and $a\!:\!(\succeq\!n\,s.C)$, where $s$ is a numeric role. The difference between $a\!:\!(\preceq\!n\,s.C)$ and $a\!:\!(\leq\!n\,s.C)$ is that, for checking $a\!:\!(\preceq\!n\,s.C)$, we do not have to pay attention to assertions of the form $s(a,b)$ or $r(a,b)$ with $r$ being a subrole of $s$. The aim for $a\!:\!(\succeq\!n\,s.C)$ is similar. We use $a\!:\!(\preceq\!n\,s.C)$ and $a\!:\!(\succeq\!n\,s.C)$ only as syntactic representations of some expressions, and do not provide semantics for them. 
We define 
\[
\FullLabel(v) = \Label(v) \cup \RFormulas(v) - \{\textrm{formulas of the form } a\!:(\preceq\!n\,s.C) \textrm{ or } a\!:(\succeq\!n\,s.C)\}.
\]

We apply global caching: if $v_1,v_2 \in V$, $\Label(v_1) = \Label(v_2)$ and ($\SType(v_1) = \SType(v_2) = \Simple$ or ($\SType(v_1) = \SType(v_2) = \Complex$ and $\Type(v_1) = \Type(v_2)$)) then $v_1 = v_2$. 
Due to global caching, an edge outgoing from a state may have a number of labels as the result of merging edges.
\LongVersion{Creation of a new node or a new edge is done by Procedure $\ConToSucc$ (connect to a successor) given on page~\pageref{proc: ConToSucc}. This procedure creates a connection from a node $v$ given as the first parameter to a node $w$ with $\Type(w)$, $\SType(w)$, $\Label(w)$, $\RFormulas(w)$, $\Repl(w)$, $\ELabels(v,w)$ specified by the remaining parameters. 


\begin{function}[t]
\caption{ConToSucc($v, type, sType, label, rFmls, indRepl, eLabel$)\label{proc: ConToSucc}}
\GlobalData{a rooted graph $\tuple{V,E,\nu}$.}
\Purpose{connect a node $v$ to a successor, which is created if necessary.}

\uIf{$v \neq \Null$ and there exists $w \in V$ such that $\Label(w) = label$ and $(\Type(w) = type$ or $\SType(w) = \Simple)$}{
   $E := E \cup \{\tuple{v,w}\}$,\ \ 
   $\RFormulas(w) := \RFormulas(w) \cup rFmls$\;
   \lIf{$\Type(v) = \State$}{$\ELabels(v,w) := \ELabels(v,w) \cup \{eLabel\}$}\;
}
\Else{
   create a new node $w$,\ \ 
   $V := V \cup \{w\}$,\ \ 
   \lIf{$v \neq \Null$}{$E := E \cup \{\tuple{v,w}\}$}\;

   $\Type(w) := type$,\ 
   $\SType(w) := sType$,\ 
   $\Status(w) := \Unexpanded$\; 
   $\Label(w) := label$,\ 
   $\RFormulas(w) := rFmls$\; 

   \lIf{$type = \State$}{$\ILConstraints(w) := \emptyset$}\;
   \lIf{$v \neq \Null$ and $\Type(v) = \State$}{$\ELabels(v,w) := \{eLabel\}$}\;
   \lIf{$indRepl \neq \Null$}{$\Repl(w) := indRepl$}\\
   \ElseIf{$sType = \Complex$}{
	\lForEach{individual $a$ occurring in $\Label(w)$}{$\Repl(w)(a) := a$}\;
   }
}

\Return{$w$}\;
\end{function}
} 

We say that a node $v$ may affect the status of the root $\nu$ if there exists a path consisting of nodes $v_0 = \nu, v_1, \ldots, v_{n-1}, v_n = v$ such that, for every $0 \leq i < n$, $\Status(v_i)$ differs from $\Sat$ and $\Unsat$, and if it is $\UnsatWrt(U)$ then $U$ is disjoint from $\{v_0,\ldots,v_i\}$. In that case, if $u \in \{v_1,\ldots,v_n\}$ then we say that $v$ may affect the status of the root~$\nu$ via a path through~$u$.


From now on, let $\tuple{\mR,\mT,\mA}$ be a knowledge base in NNF of the logic \SHOQ, with $\mA \neq \emptyset$.$\,$\footnote{If $\mA$ is empty, we can add $a\!:\!\top$ to it, where $a$ is a special individual.} In this section we present a tableau calculus \CSHOQ for checking satisfiability of $\tuple{\mR,\mT,\mA}$.
A~\CSHOQ-tableau for $\tuple{\mR,\mT,\mA}$ is a rooted graph $G = \tuple{V,E,\nu}$ constructed as follows.

\LongVersion{\subsubsection*{Initialization:} 
Set $V := \emptyset$ and $E := \emptyset$. Then, create the root node by executing $\nu := \ConToSucc(\Null, \NonState, \Complex, label, \emptyset, \Null, \Null)$, where $label = \mA\ \cup$ $\{(a\!:\!C) \mid$ $C \in \mT$ and $a$ is an individual occurring in $\mA$ or $\mT\}$.}
\ShortVersion{

\bigskip

\noindent{\bf Initialization:} 
$V := \{\nu\}$, 
$E := \emptyset$, 
$\Type(\nu) := \NonState$, 
$\SType(\nu) := \Complex$, 
$\Status(\nu) := \Unexpanded$, 
$\RFormulas(\nu) := \emptyset$, 
$\Label(\nu) := \mA\ \cup$ $\{(a\!:\!C) \mid$ $C \in \mT$ and $a$ is an individual occurring in $\mA$ or $\mT\}$, 
for each individual $a$ occurring in $\Label(\nu)$ set $\Repl(\nu)(a) := a$.
}

\LongVersion{\subsubsection*{Rules' Priorities and Expansion Strategies:}}
\ShortVersion{

\bigskip

\noindent{\bf Rules' Priorities and Expansion Strategies:}}
The graph is then expanded by the following rules, 
\LongVersion{which will be specified shortly:}
\ShortVersion{which are specified in detail in~\cite{SHOQ-long}:}

\begin{tabbing}
mm \= mmmm \= \kill
\> \UPS \> rules for updating statuses of nodes,\\[0.5ex]
\> \US \> unary static expansion rules,\\[0.5ex]
\> \DN \> a rule for dealing with nominals,\\[0.5ex]
\> \NUS \> a non-unary static expansion rule,\\[0.5ex]
\> \FS \> the forming-state rule,\\[0.5ex]
\> \TP \> a transitional partial-expansion rule,\\[0.5ex]
\> \TF \> a transitional full-expansion rule.
\end{tabbing}

Each of the rules is parametrized by a node $v$. We say that a rule is {\em applicable} to $v$ if it can be applied to $v$ to make changes to the graph. 
The rule \UPS has a higher priority than \US, which has a higher priority than the remaining rules in the list. 
If neither \UPS nor \US is applicable to any node, then choose a node $v$ with status $\Unexpanded$ or $\PExpanded$, choose the first rule applicable to $v$ among the rules in the last five items of the above list, and apply it to~$v$. Any strategy can be used for choosing $v$, but it is worth to choose $v$ for expansion only when $v$ may affect the status of the root $\nu$ of the graph. 
Note that the priorities of the rules are specified by the order in the above list, but the rules \UPS and \US are checked globally (technically, they are triggered immediately when possible), while the remaining rules are checked for a chosen node.

\ShortVersion{

\bigskip

\noindent{\bf Termination:} }
The construction of the graph ends when the root $\nu$ receives the status $\Unsat$ or $\Sat$ or when no more changes that may affect the status of $\nu$ can be made.\footnote{That is, ignoring nodes that are unreachable from $\nu$ via a path without nodes with status $\Unsat$ or $\Sat$, no more changes can be made to the graph.} Theorem~\ref{theorem: s-c} states that the knowledge base $\tuple{\mR,\mT,\mA}$ is satisfiable iff $\Status(\nu) \neq \Unsat$.


\subsection{Techniques for Dealing with Nominals}

As usual, to deal with assertions of the form $a\!:\!\{b\}$ (resp.\ $a\!:\!\lnot\{b\}$) we use the predicate $\doteq$ (resp.~$\not\doteq$) and the replacement technique. 
Recall also that we use statuses of the form $\UnsatWrt(U)$ in order to be able to apply global caching in the presence of nominals. Updating statuses of nodes is defined appropriately\ShortVersion{ (see~\cite{SHOQ-long})}. 
Our other techniques for dealing with nominals are described below. 

Suppose $v$ is a simple node with $\Status(v) \notin \{\Unsat,\Sat\}$ and $\{a\} \in \Label(v)$, a complex state~$u$ is an ancestor of $v$, and $v$ may affect the status of the root $\nu$ via a path through~$u$. Let $u_0$ be a predecessor of $u$. The node $u_0$ has only $u$ as a successor and it was expanded by the forming-state rule. There are three cases:
\begin{itemize}
\item If, for every $C \in \Label(v)$, the formula obtained from $a\!:\!C$ by replacing every individual $b$ with $\Repl(u)(b)$, when $\Repl(u)(b)$ is defined, belongs to $\FullLabel(u)$, then $v$ is ``consistent'' with $u$. 
\item If there exists $C \in \Label(v)$ such that the formula obtained from $a\!:\!\ovl{C}$ (where $\ovl{C}$ is the negation of $C$ in NNF) by replacing every individual $b$ with $\Repl(u)(b)$, when $\Repl(u)(b)$ is defined, belongs to $\FullLabel(u)$, then $v$ is ``inconsistent'' with $u$. In this case, if $\Status(v)$ is of the form $\UnsatWrt(U)$ then we update it to $\UnsatWrt(U \cup \{u\})$, else we update it to $\UnsatWrt(\{u\})$. 
\item In the remaining case, the node $u$ is ``incomplete'' w.r.t.\ $v$, which means that the expansion of $u_0$ was not appropriate. Thus, we delete the edge $\tuple{u_0,u}$ and re-expand $u_0$ by an appropriate ``or''-branching\ShortVersion{ (see~\cite{SHOQ-long})}.
\end{itemize}

For dealing with interaction between number restrictions and nominals, to guarantee that every nominal represents a singleton set and a named individual cannot be cloned we use concepts of the form $\leq\!1\,r.\{a\}$ and assertions of the form $a\!:\,\leq\!1\,r.\{b\}$ that are {\em relevant} w.r.t.\ the TBox $\mT$ and the label of the considered node. One can define this relation to be the full one (i.e., so that such formulas are always relevant). However, to increase efficiency we define this notion as follows. 

Let $X$ be a set of formulas. We say that a formula $\varphi$ {\em occurs positively at the modal depth $0$ in $X$} if there exist $\psi \in X$ and an occurrence of $\varphi$ in $\psi$ that is not in the scope of $\lnot$, $\E r$, $\V r$, $\geq\!r\,n$, $\leq\!r\,n$ for any $r \in \RN$. (Recall that formulas are in NNF.)

\begin{definition}\LongVersion{\em}
We say that a concept $\leq\!1\,r.\{a\}$ is {\em relevant} w.r.t.\ a TBox $\mT$ and a set $X$ of concepts if the following conditions hold:
\begin{itemize}
\item either of the following conditions holds:
   \begin{itemize}
   \item some concept $\geq\!m\,s.C$ with $m \geq 2$ and $s \sqsubseteq_\mR r$ occurs positively at the modal depth $0$ in $X$ -- in this case, let $s_1 = s_2 = s$ and $C_1 = C_2 = C$,
   \item some different concepts $C'_1$ and $C'_2$ occur positively at the modal depth $0$ in $X$ and each $C'_i$ is of the form $\E s_i.C_i$ or $\geq\!1\,s_i.C_i$ with $s_i \sqsubseteq_\mR r$; 
   \end{itemize}
\item one of the following conditions holds:
   \begin{itemize}
   \item the nominal $\{a\}$ occurs positively at the modal depth $0$ in $\mT$ or $\{C_1\}$ or $\{C_2\}$, 
   \item some concept $\V r'.D$ satisfying $s_1 \sqsubseteq r'$ and $s_2 \sqsubseteq r'$ occurs positively at the modal depth~$0$ in $\mT$ and the nominal $\{a\}$ occurs positively at the modal depth $0$ in $\{D\}$,
   \item some concept $\leq\!n\,r'.D$ satisfying $s_1 \sqsubseteq r'$ and $s_2 \sqsubseteq r'$ occurs positively at the modal depth~$0$ in $\mT$ and either $\{a\}$ or $\lnot\{a\}$ occurs positively at the modal depth $0$ in $\{D\}$.
   \end{itemize}
\end{itemize}

A concept $a\!:\,\leq\!1\,r.\{b\}$ is {\em relevant} w.r.t.\ a TBox $\mT$ and a set $X$ of eABox assertions if the concept $\leq\!1\,r.\{b\}$ is relevant w.r.t.\ $\mT$ and the set $\{C \mid a\!:\!C \in X\}$.
\myEnd
\end{definition}

Note that every interpretation $\mI$ always validates $\leq\!1\,r.\{a\}$ (i.e., $(\leq\!1\,r.\{a\})^\mI = \Delta^\mI$) and satisfies $a\!:\,\leq\!1\,r.\{b\}$ (i.e., $\mI \models a\!:\,\leq\!1\,r.\{b\}$). 


\LongVersion{
\input{SHOQ-examples}

\subsection{Tableau Rules}

In this subsection we formally specify the tableau rules of our calculus \CSHOQ. We also give explanations for them. They are informal and should be understood in the context of the described rule. 

We will use the auxiliary procedure $\SetClosedWrt(v,u)$ defined as follows: ``if $\Status(v)$ is of the form $\UnsatWrt(U)$ then $\Status(v) := \UnsatWrt(U \cup \{u\})$, else $\Status(v) := \UnsatWrt(\{u\})$''. This procedure updates the status of $v$ to reflect that $v$ is closed w.r.t.~$u$. 

\sssection{The Rules for Updating Statuses of Nodes:} 

\begin{description}
\item[\UPSa] The first rule is as follows:
  \begin{enumerate}
  \item\label{item: UFRMS} if one of the following conditions holds:
	\begin{enumerate}
	\item there exists $\alpha\!:\!\bot \in \Label(v)$ or $\{\varphi,\ovl{\varphi}\} \subseteq \FullLabel(v)$,
	\item there exists $a \not\doteq a \in \Label(v)$,
	\item $\Type(v) = \NonState$, $a\!:\!(\leq\!n\,s.C) \in \Label(v)$ and there are $b_0,\ldots,b_n \in \IN$ such that, for all $0 \leq i,j \leq n$ with $i \neq j$, $\{s(a,b_i), b_i\!:\!C\} \subseteq \FullLabel(v)$ and $\{b_i \not\doteq b_j$, $b_j \not\doteq b_i\} \cap \Label(v) \neq \emptyset$, 
	\item $\Status(v) = \UnsatWrt(U)$ and $v \in U$, 
	\end{enumerate}
   then $\Status(v) := \Unsat$ 
  \item else if $\Type(v) = \State$, $\Status(v) = \FExpanded$ and $v$ has no successors then $\Status(v) := \Sat$.
  \end{enumerate}

  \begin{explanation}
  Informally, $\Unsat$ means ``unsatisfiable w.r.t.\ $\mR$ and $\mT$'', $\Sat$ means ``satisfiable w.r.t.\ $\mR$ and $\mT$'', and $\UnsatWrt(U)$ means ``unsatisfiable w.r.t.\ $\mR$, $\mT$ and any node from $U$''. The above rule is thus intuitive. For a formal characterization of the statuses $\Unsat$ and $\UnsatWrt(U)$, we refer the reader to Lemma~\ref{lemma: SHQWD} (on page~\pageref{lemma: SHQWD}).
\koniec
  \end{explanation}

\item[\UPSb] 
  If $\SType(v) = \Simple$, $\Status(v) \notin \{\Unsat,\Sat\}$ and $\Label(v)$ contains a concept $\{a\}$ then
     \begin{enumerate}
     \item[] if there exist $u \in V$ and $C \in \Label(v)$ such that 
	\begin{itemize} 
	\item $\Type(u) = \State \land \SType(u) = \Complex$,
	\item $v$ may affect the status of the root $\nu$ via a path through $u$,
	\item the assertion obtained from $a\!:\!\ovl{C}$ by replacing every individual $b$ by $\Repl(u)(b)$ when $\Repl(u)(b)$ is defined belongs to $\FullLabel(u)$ 
	\end{itemize}
	then $\SetClosedWrt(v,u)$.
     \end{enumerate}

  \begin{explanation}
  This rule deals with nominals. 
  If $u$ is a complex state and $v$ is a simple node such that $\Label(v)$ contains a nominal $\{a\}$ and $v$ may affect the status of the root $\nu$ via a path through $u$ then, when considering $u$ for constructing a model for the considered knowledge base, the simple node $v$ should be merged with the named individual $a$ in the complex node $u$. If such a merging causes inconsistency then $v$ is closed w.r.t.\ $u$ and we update $\Status(v)$ accordingly. 
\koniec
  \end{explanation}

\item[\UPSc] This rule states that, if $v$ is a predecessor of a node $w$ then, whenever the status of $w$ changes to $\Unsat$, $\UnsatWrt(\ldots)$ or $\Sat$, the status of $v$ should be updated (as soon as possible by using a priority queue of tasks). The update is done by one of the following subrules: 
  \begin{enumerate}
  \item If $\Type(v) = \NonState$ and $\Status(v) \notin \{\Unexpanded,\Unsat,\Sat\}$ then
     \begin{enumerate}
     \item\label{item: IUQKD} if some successor of $v$ received status $\Sat$ then $\Status(v) := \Sat$
     \item\label{item: YUSAO} else if all successors of $v$ have status $\Unsat$ then $\Status(v) := \Unsat$
     \item\label{item: YUSLO} else if every successor of $v$ has status $\Unsat$ or $\UnsatWrt(\ldots)$ then:
	\begin{enumerate}
	\item let $w_1,\ldots,w_k$ be all the successors of $v$ such that, for $1 \leq i \leq k$, $\Status(w_i)$ is of the form $\UnsatWrt(U_i)$, and let $U = \bigcap_{1 \leq i \leq k} U_i$
	\item for each $u \in U$ do $\SetClosedWrt(v,u)$.
	\end{enumerate}
     \end{enumerate}

  \item If $\Type(v) = \State$, $\Status(v) \notin \{\Unexpanded,\Unsat,\Sat\}$ and a successor $w$ of $v$ received status $\Unsat$ then 
	\begin{enumerate}
	\item\label{item: YTDSA} if there exists $e \in \ELabels(v,w)$ such that $\piT(e) = \TUS$\\ then $\Status(v) := \Unsat$
	\item else\label{item: HGWSS}
	  \begin{itemize}
	  \item for each $e \in \ELabels(v,w)$ such that $\piT(e) = \CQF$ do\\ 
		\mbox{\hspace{1em}}add the constraint $x_{w,e} = 0$ to $\ILConstraints(v)$
	  \item if $\ILConstraints(v)$ is infeasible then $\Status(v) := \Unsat$.
	  \end{itemize}
	\end{enumerate}

  \item If $\Type(v) = \State$, $\Status(v) \notin \{\Unexpanded,\Unsat,\Sat\}$, a successor $w$ of $v$ received status $\UnsatWrt(U)$, and $v$ may affect the status of the root $\nu$ via a path through $u \in U$ then
	\begin{enumerate}
	\item\label{item: IDSMG} if there exists $e \in \ELabels(v,w)$ such that $\piT(e) = \TUS$\\ then $\SetClosedWrt(v,u)$
	\item else\label{item:IRDMS}
	  \begin{itemize}
	  \item let $\tuple{w_1,e_1},\ldots,\tuple{w_k,e_k}$ be all the pairs such that, for $1 \leq i \leq k$, $w_i$ is a successor of $v$, $\Status(w_i)$ is of the form $\UnsatWrt(U_i)$ with $u \in U_i$, $e_i \in \ELabels(v,w_i)$, and $\piT(e_i) = \CQF$
 	  \item if $\ILConstraints(v) \cup \{x_{w_i,e_i} = 0 \mid 1 \leq i \leq k\}$ is infeasible\\ then $\SetClosedWrt(v,u)$.
	  \end{itemize}
	\end{enumerate}

  \item\label{item: HFDDE} If
	\begin{itemize}
	\item $\Type(v) = \State \land \Status(v) = \FExpanded$, 
	\item every successor $w$ of $v$ with some $e \in \ELabels(v,w)$ having $\piT(e) = \TUS$ has status $\Sat$, and 
	\item $\ILConstraints(v) \cup \{x_{w,e} = 0 \mid$ $\tuple{v,w} \in E$, $\Status(w) \neq \Sat$, $e \in \ELabels(v,w)$ and $\piT(e) = \CQF\}$ is feasible 
	\end{itemize}
     then $\Status(v) := \Sat$.
  \end{enumerate}

  \begin{explanation}
  For simplicity of understanding, one can first consider the case without nominals and statuses $\UnsatWrt(\ldots)$. 
  A non-state is like an ``or''-node, whose status is the disjunction of the statuses of its successors, treating $\Sat$ as $\True$ and $\Unsat$ as $\False$. 
  A state is more sophisticated than an ``and''-node. The status of a state $v$ is different from $\Unsat$ iff the following conditions hold:
    \begin{itemize}
    \item for all successors $w$ of $v$, if there exists $e \in \ELabels(v,w)$ with $\piT(e) = \TUS$ then $\Status(w) \neq \Unsat$,
    \item $\ILConstraints(v) \cup \{x_{w,e} = 0 \mid$ $\tuple{v,w} \in E$, $\Status(w)=\Unsat$, $e \in \ELabels(v,w)$ and $\piT(e) = \CQF\}$ is feasible.
    \end{itemize}

The subrule~2 updates the status of a state $v$ according to the above observation. The status $\Sat$ is a special case of being different from $\Unsat$, which can be detected earlier, and the subrule~4 is defined appropriately, reflecting that observation. 

\smallskip

Recall that $\UnsatWrt(U)$ means $\Unsat$ w.r.t.\ any node $u \in U$, and such statuses are used for dealing with nominals. In the case with nominals and statuses $\UnsatWrt(\ldots)$, for simplicity of understanding, one can imagine the traditional approach that does not use global caching but uses backtracking for dealing with ``or''-branchings. With that approach, each node has at most one ancestor node $u$ that is a complex state, and a status $\UnsatWrt(\ldots)$ behaves similarly to the status $\Unsat$. Our approach uses global caching and deals with nominals, and we use statuses $\UnsatWrt(\ldots)$ in appropriate way to simulate the status $\Unsat$. 
\koniec
  \end{explanation}

\end{description}

\sssection{The Unary Static Expansion Rules:}
  \begin{description}
  \item[\USa] If $\Type(v) = \NonState$ and $\Status(v) = \Unexpanded$ then 
    \begin{enumerate}
     \item let $X = \RFormulas(v) \cup \{(\alpha\!:\!C) \in \Label(v) \mid C$ is of the form $D \mand D'$ or $\geq\!0\,s.D$ or $\leq\!0\,s.D\} \cup \{a\!:\!\lnot\{b\} \in \Label(v)\}$ 
     \item let $label = \Label(v) \cup \{(\alpha\!:\!D), (\alpha\!:\!D') \mid \alpha\!:\!(D \mand D') \in \Label(v)\}$ \\
	\mbox{\hspace{4.5em}} $\cup\ \{\alpha\!:\!\V s.\ovl{D} \mid\ (\alpha\!:\,\leq\!0\,s.D) \in \Label(v)\}$ \\
	\mbox{\hspace{4.5em}} $\cup\ \{\alpha\!:\!\V r.D \mid \alpha\!:\!\V s.D \in \Label(v) \textrm{ and } r \sqsubseteq_\mR s\}$ \\
	\mbox{\hspace{4.5em}} $\cup\ \{s(a,b) \mid r(a,b) \in \Label(v) \textrm{ and } r \sqsubseteq_\mR s\}$ \\
	\mbox{\hspace{4.5em}} $\cup\ \{b\!:\!D \mid \{a\!:\!\V r.D, r(a,b)\} \subseteq \Label(v)\}$ \\
	\mbox{\hspace{4.5em}} $\cup\ \{b\!:\!\V r.D \mid \{a\!:\!\V r.D, r(a,b)\} \subseteq \Label(v) \textrm{ and } \Trans{r}\}$ \\
	\mbox{\hspace{4.5em}} $\cup\ \{a \not\doteq b \mid a\!:\!\lnot\{b\} \in \Label(v)\}$ \\
	\mbox{\hspace{4.5em}} $-\ X$
     \item if $label - \Label(v) \neq \emptyset$ then
	\begin{enumerate}
	\item $\ConToSucc(v,\NonState,\SType(v),label,X,\Repl(v),\Null)$
	\item $\Status(v) := \FExpanded$.
	\end{enumerate}
     \end{enumerate}

  \begin{explanation}
  This rule makes a necessary expansion for a non-state $v$ by connecting it to only one successor $w$ which is a copy of $w$ with intuitive changes like:
  \begin{itemize}
  \item if $\alpha\!:\!(D \mand D') \in \Label(v)$ then $\alpha\!:\!(D \mand D')$ in $\Label(w)$ is replaced by $\alpha\!:\!D$ and $\alpha\!:\!D'$ and we remember this by adding it to $\RFormulas(w)$;
  \item if $\{a\!:\!\V r.D, r(a,b)\} \subseteq \Label(v)$ then we add $b\!:\!D$ to $\Label(w)$; and so on.
  \end{itemize}
  Note that $\Label(w) - (\Label(v) \cup \RFormulas(v)) \neq \emptyset$. That is, $w$ contains some ``new'' formulas. 
\koniec 
  \end{explanation}

  \item[\USb] If $\Status(v) = \Unexpanded$ and $\Label(v)$ contains $a\!:\!\{b\}$ then
     \begin{enumerate}
     \item let $X$ be the set obtained from $\Label(v) - \{a\!:\!\{b\}\}$ by replacing every occurrence of $b$ not in $\doteq$~expressions by~$a$
     \item let $Y$ be the set obtained from $\RFormulas(v)$ by replacing every occurrence of $b$ by~$a$
     \item $w := \ConToSucc(v,\NonState,\Complex,X \cup \{a \doteq b, b \doteq a\},Y \cup \{a\!:\!\{a\}\},\Repl(v),\Null)$
     \item $\Repl(w)(b) := a$
     \item for each $c \in \IN$, if $\Repl(v)(c) = b$ then $\Repl(w)(c) := a$.
     \item $\Status(v) := \FExpanded$.
     \end{enumerate}
  \begin{explanation}
  If $v$ is an unexpanded complex node with $\Label(v)$ containing $a\!:\!\{b\}$ then $a$ and $b$ should denote the same individual and we expand $v$ by connecting it to only one successor $w$ which is a copy of $v$ with $b$ replaced by $a$ in an appropriate way. 
\koniec 
  \end{explanation}

  \item[\USc] If $\Type(v) = \NonState$ and $\Status(v) = \Unexpanded$ then 
     \begin{enumerate}
     \item if $\SType(v) = \Simple$ then let $X$ be the set of all concepts of the form $\leq\!1\,r.\{a\}$ that are relevant w.r.t.\ $\mT$ and $\Label(v)$, else let $X$ be the set of all formulas of the form $a\!:\,\leq\!1\,r.\{b\}$ that are relevant w.r.t.\ $\mT$ and $\Label(v)$
     \item if $X - \Label(v) \neq \emptyset$ then\\
	$\ConToSucc(v,\NonState,\SType(v),\Label(v) \cup X,\RFormulas(v),\Repl(v),\Null)$.
     \end{enumerate}
   
  \begin{explanation}\label{exp: HGALC}
This rule deals with interaction between number restrictions and nominals. We want to guarantee that every nominal represents a singleton set and a named individual cannot be cloned.
\koniec 
  \end{explanation}
  \end{description}

\sssection{The Rule for Dealing with Nominals:}
  \begin{description}
  \item[\DN] If $\SType(v) = \Simple$, $\Status(v) \notin \{\Unsat,\Sat\}$ and $\Label(v)$ contains $\{a\}$ then
     \begin{enumerate}
     \item for each complex state $u$ such that 
	$v$ is not closed w.r.t. $u$ (i.e., $\Status(v)$ is not of the form $\UnsatWrt(U)$ with $u \in U$) 
	and $v$ may affect the status of the root $\nu$ via a path through $u$, 
	do
	\begin{enumerate}
	\item\label{step: EROSA} let $X = \{\varphi_1,\ldots,\varphi_k\}$ be the set obtained from $\{a\!:\!C \mid C \in \Label(v)\}$ by replacing every individual $b$ by $\Repl(u)(b)$ when $\Repl(u)(b)$ is defined
	\item if $X \nsubseteq \FullLabel(u)$ then:\\ for each predecessor $u_0$ of $u$ do
	   \begin{enumerate}
	   \item delete the edge $\tuple{u_0,u}$ from $E$ and its labels from $\ELabels$
	   \item $\ConToSucc(u_0$, $\NonState$, $\Complex$, $\Label(u_0) \cup X$, $\RFormulas(u_0)$, $\Repl(u_0)$, $\Null)$
	   \item for each $1 \leq i \leq k$ such that $\varphi_i \notin \FullLabel(u)$ do:\\
		$\ConToSucc(u_0$, $\NonState$, $\Complex$, $\Label(u_0) \cup \{\ovl{\varphi_i}\}$, $\RFormulas(u_0)$, $\Repl(u_0)$, $\Null)$
	   \end{enumerate}
	\end{enumerate}
     \item $\Status(v) := \Blocked$.
     \end{enumerate}

  \begin{explanation}
  If $u$ is a complex state and $v$ is a simple node such that $\Label(v)$ contains a nominal $\{a\}$ and $v$ may affect the status of the root $\nu$ via a path through $u$ then, when considering $u$ for constructing a model for the considered knowledge base, the simple node $v$ should be merged with the named individual $a$ in the complex node $u$. The case when such a merging causes inconsistency is dealt with by the rule~\UPSb (with a higher priority). Consider the other case. The set $X = \{\varphi_1,\ldots,\varphi_k\}$ of assertions computed at the step~\ref{step: EROSA} of the rule would be added to $\Label(u)$. However, we do not want to modify labels of nodes. In the case when $X \nsubseteq \FullLabel(u)$, the label of $u$ is ``incomplete'' and we re-expand every predecessor $u_0$ of $u$ by deleting the edge $\tuple{u_0,u}$ and connecting $u_0$ to $k+1$ successors, where the label of the successor number 0 extends $\Label(u_0)$ with $X = \{\varphi_1,\ldots,\varphi_k\}$ and the label of the successor number $i$ ($1 \leq i \leq k$) extends $\Label(u_0)$ with $\ovl{\varphi_i}$ (the negation of $\varphi_i$ in NNF). This is like an on-demand cut.
\koniec
  \end{explanation}
  \end{description}

\sssection{The Non-unary Static Expansion Rule:}
  \begin{description}
  \item[\NUS] If $\Type(v) = \NonState$ and $\Status(v) = \Unexpanded$ then
     \begin{enumerate}
     \item\label{item: JHREA} if $\alpha\!:\!(C \mor D) \in \Label(v)$ and $\{\alpha\!:\!C, \alpha\!:\!D\} \cap \FullLabel(v) = \emptyset$ then
	\begin{enumerate}
	\item let $X = \Label(v) - \{\alpha\!:\!(C \mor D)\}$ 
	\item let $Y = \RFormulas(v)\cup\{\alpha\!:\!(C \mor D)\}$
	\item $\ConToSucc(v,\NonState,\SType(v),X\cup\{\alpha\!:\!C\},Y,\Repl(v),\Null)$
	\item $\ConToSucc(v,\NonState,\SType(v),X\cup\{\alpha\!:\!D\},Y,\Repl(v),\Null)$
	\item $\Status(v) := \FExpanded$
	\end{enumerate}

     \begin{explanation}
     This subrule deals with syntactic branching on $\alpha\!:\!(C \mor D) \in \Label(v)$. We expand $v$ by connecting it to two successors $w_1$ and $w_2$, whose labels are the label of $v$ with $\alpha\!:\!(C \mor D)$ replaced by $\alpha\!:\!C$ or $\alpha\!:\!D$, respectively. The formula $\alpha\!:\!(C \mor D)$ is put into both $\RFormulas(w_1)$ and $\RFormulas(w_2)$. The expansion is done only when both $w_1$ and $w_2$ have a larger $\FullLabel$ than $v$. 
\koniec
     \end{explanation}

     \item\label{item: HGW3A} else if $\SType(v) = \Complex$, $s(a,b) \in \Label(v)$ and 
       \begin{itemize}
       \item $\Label(v)$ contains $a\!:\!(\leq\!n\,s.C)$, or 
       \item $\Label(v)$ contains $a\!:\!(\geq\!n\,s.C)$ or $a\!:\!(\E s.C)$, where $s$ is a numeric role, 
       \end{itemize}
     and $\{b\!:\!C, b\!:\!\ovl{C}\} \cap \FullLabel(v) = \emptyset$ then 
	\begin{enumerate}
	\item let $X = \Label(v)\cup\{b\!:\!C\}$ and $X' = \Label(v)\cup\{b\!:\!\ovl{C}\}$
	\item $\ConToSucc(v, \NonState, \Complex, X, \RFormulas(v), \Repl(v), \Null)$
	\item $\ConToSucc(v, \NonState, \Complex, X', \RFormulas(v), \Repl(v), \Null)$
	\item $\Status(v) := \FExpanded$
	\end{enumerate}

     \begin{explanation}
  This subrule deals with the case when there is a lack of information about $b$ for deciding how to satisfy the number restrictions about $a$. We want to have either $b\!:\!C$ or $b\!:\!\ovl{C}$ in $\FullLabel(v)$. So, we expand $v$ by semantic branching: we connect it to two successors, one with label $\Label(v)\cup\{b\!:\!C\}$ and the other with label \mbox{$\Label(v)\cup\{b\!:\!\ovl{C}\}$}. The expansion is done only when both the successors have a larger $\FullLabel$ than~$v$.
\koniec
     \end{explanation}

     \item\label{item: JHEAA} else if $\SType(v) = \Complex$, $\{a\!:\!(\leq\!n\,s.C)$, $s(a,b)$, $s(a,b')$, $b\!:\!C$, $b'\!:\!C\} \subseteq$ $\FullLabel(v)$, $b \neq b'$ and $\{b \not\doteq b', b' \not\doteq b\} \cap \Label(v) = \emptyset$ then\footnote{Fix a linear order between individual names. Then we can also assume that $b$ is less than $b'$ in that order.}
	\begin{enumerate}
	\item let $X_1 = \Label(v)\cup\{b \not\doteq b',b' \not\doteq b\}$ and let $X_2$ be the set obtained from $\Label(v)$ by replacing every occurrence of $b'$ not in $\doteq$~expressions by $b$
	\item let $Y$ be the set obtained from $\RFormulas(v)$ by replacing every occurrence of $b'$ by $b$
	\item $\ConToSucc(v, \NonState, \Complex, X_1, \RFormulas(v), \Repl(v), \Null)$
	\item\label{item: HJSDA} $w := \ConToSucc(v,\NonState,\Complex,X_2 \cup \{b \doteq b', b' \doteq b\},Y,\Repl(v),\Null)$
	\item $\Repl(w)(b') := b$
	\item for each $c \in \IN$, if $\Repl(v)(c) = b'$ then $\Repl(w)(c) := b$
	\item $\Status(v) := \FExpanded$
	\end{enumerate}

     \begin{explanation}
  This subrule deals with the case when there is a lack of information about whether $b$ and $b'$ denote the same individual for deciding how to satisfy the number restrictions about $a$. We expand $v$ by semantic branching: either $b$ and $b'$ denote the same individual or they do not. Technically, we connect $v$ to two successors with appropriate contents. 
\koniec
     \end{explanation}

     \item\label{item: OSJRS} else if $\SType(v) = \Complex$, $\{a\!:\!(\leq\!m\,r.C)$, $r(a,b)\} \subseteq \Label(v)$, $\Label(v)$ contains \mbox{$a\!:\!(\geq\!n\,s.D)$} or $a\!:\!\E s.D$ with $s \sqsubseteq_\mR r$, and $\{s(a,b), \lnot s(a,b)\} \cap \Label(v) = \emptyset$ then
	\begin{enumerate}
	\item let $X_1 = \Label(v) \cup \{s(a,b)\}$ and $X_2 = \Label(v) \cup \{\lnot s(a,b)\}$  
	\item $\ConToSucc(v,\NonState,\Complex,X_1,\RFormulas(v),\Repl(v),\Null)$
	\item $\ConToSucc(v,\NonState,\Complex,X_2,\RFormulas(v),\Repl(v),\Null)$
	\item $\Status(v) := \FExpanded$.
	\end{enumerate}

     \begin{explanation}
  This subrule deals with the case when there is a lack of information for deciding how to satisfy the number restrictions about $a$. We want to decide whether $b$ is an $s$-successor of $a$ or not. So, we expand $v$ by semantic branching: we connect it to two successors, one with label containing $s(a,b)$ and the other with label containing $\lnot s(a,b)$. The expansion is done only when both the successors have a larger label than~$v$.
\koniec
     \end{explanation}
     \end{enumerate}
  \end{description}

\sssection{The Forming-State Rule:}
\begin{description}
\item[\FS] If $\Type(v) = \NonState$ and $\Status(v) = \Unexpanded$ then 
   \begin{enumerate}
   \item if $\SType(v) = \Simple$ then $\Type(v) := \State$ and $\ILConstraints(v) := \emptyset$
   \item else 
     \begin{enumerate}
     \item set $X := \Label(v)$
     \item for each $a\!:\!(\leq\!n\,s.D) \in \Label(v)$ do
	\begin{enumerate}
	\item let $m = \sharp\{b \mid \{s(a,b), b\!:\!D\} \subseteq \FullLabel(v)\}$
	\item add $a\!:\!(\preceq\!(n-m)\,s.D)$ to $X$
	\end{enumerate}
     \item for each $(a\!:\!C) \in \Label(v)$, where $C$ is $\geq\!n\,s.D$ or $\E s.D$ and $s$ is a numeric role, do
	\begin{enumerate}
	\item if $C = \E s.D$ then let $n = 1$
	\item let $m = \sharp\{b \mid \{s(a,b), b\!:\!D\} \subseteq \FullLabel(v)\}$
	\item if $n > m$ then add $a\!:\!(\succeq\!(n-m)\,s.D)$ to $X$
	\end{enumerate}
     \item $\ConToSucc(v,\State,\Complex,X,\RFormulas(v),\Repl(v),\Null)$
     \item $\Status(v) := \FExpanded$.
     \end{enumerate}
   \end{enumerate}

\begin{explanation}
When the rules \UPS, \US, \DN and \NUS are not applicable to the non-state $v$, we apply this forming-state rule to $v$. If $v$ is a complex node then we connect it to a complex state $w$. When computing contents for $w$ we put into $\Label(w)$ the requirements from $\Label(v)$ after an appropriate modification that takes into account the assertions in $\Label(v)$ that represent the relationship between named individuals. For example, if \mbox{$a\!:\!(\leq\!n\,s.D) \in \Label(v)$} and there are $m$ pairwise different individuals $b_1,\ldots,b_m$ such that $\{s(a,b_i), b_i\!:\!D \mid 1 \leq i \leq m\} \subseteq \FullLabel(v)$ then we add to $\Label(w)$ the requirements $a\!:\!(\preceq\!(n-m)\,s.D)$. Notice the use of $\preceq$ instead of $\leq$. Note that, since the rule \NUS is not applicable to $v$, we must have that $(b_i \not\doteq b_j) \in \Label(v)$ for any pair $i \neq j$, and for any individual $b$ such that $s(a,b) \in \Label(v)$, either $b\!:\!D$ or $b\!:\!\ovl{D}$ must belong to $\FullLabel(v)$. When expanding $w$ we will not have to pay attention to the relationship between the individuals occurring in $\Label(w)$. 

\smallskip

If $v$ is a simple node then we just change $\Type(v)$ to $\State$ and initialize $\ILConstraints(v)$ to $\emptyset$. Number restrictions about $v$ are dealt with later by the transitional full-expansion rule.

\smallskip

The way of forming a state for a complex node $v$ is more sophisticated (than for a simple node) because we may need to re-expand $v$ later due to nominals (as done by the rule \DN for $u_0$).
\koniec
\end{explanation}
\end{description}

\sssection{The Transitional Partial-Expansion Rule:}
  \begin{description}
  \item[\TP] If $\Type(v) = \State$ and $\Status(v) = \Unexpanded$ then
     \begin{enumerate}
     \item for each $(\alpha\!:\!\E r.D) \in \Label(v)$, where $r$ is a non-numeric role, do
       \begin{enumerate}
	\item $X := \{D\} \cup \{D' \mid \alpha\!:\!\V r.D' \in \Label(v)\}\ \cup$\\
		\mbox{\hspace{2.55em}}$\{\V s.D' \mid \alpha\!:\!\V s.D' \in \Label(v), r \sqsubseteq_\mR s$ and $\Trans{s}\} \cup \mT$
	\item $eLabel := \langle \TUS, \{s \mid r \sqsubseteq_\mR s\}, \alpha\rangle$
	\item $\ConToSucc(v,\NonState,\Simple,X,\emptyset,\Null,eLabel)$
       \end{enumerate}

     \item $\Status(v) := \PExpanded$.
     \end{enumerate}
  \begin{explanation}
  To realize a requirement $\alpha\!:\!\E r.D$ at a state $v$, where $r$ is a non-numeric role, we connect $v$ to a new simple non-state $w$ with appropriate contents as shown in the rule.
\koniec
  \end{explanation}
  \end{description}

\sssection{The Transitional Full-Expansion Rule:}
  \begin{description}
  \item[\TF] If $\Type(v) = \State$ and $\Status(v) = \PExpanded$ then
     \begin{enumerate}
     \item\label{step: UIRJS} if $\SType(v) = \Complex$\\ 
	then let $\Gamma = \Label(v)$\\ 
	else let $\Gamma = \Label(v) \cup \{\preceq\!n\,r.D \mid$ $\leq\!n\,r.D \in \Label(v)\}\ \cup$\\ 
	\mbox{\hspace{5.4em}}$\{\succeq\!n\,r.D \mid$ $\geq\!n\,r.D \in \Label(v)\}\ \cup$\\ 
	\mbox{\hspace{5.4em}}$\{\succeq\!1\,r.D \mid$ $\E r.D \in \Label(v)$ and $r$ is a numeric role$\}$
     \item $\mE := \emptyset$, $\mE' := \emptyset$
     \item\label{item: HJFDU} for each $(\alpha\!:\,\succeq\!n\,r.D) \in \Gamma$ do
       \begin{enumerate}
	\item $X := \{s \mid r \sqsubseteq_\mR s\}$
	\item $Y := \{D\} \cup \{D' \mid \alpha\!:\!\V r.D' \in \Gamma\}\ \cup$\\
		\mbox{\hspace{2.55em}}$\{\V s.D' \mid \alpha\!:\!\V s.D' \in \Gamma, r \sqsubseteq_\mR s$ and $\Trans{s}\} \cup \mT$
	\item $\mE := \mE \cup \{\tuple{X,Y,\alpha}\}$
       \end{enumerate}
     \item\label{item: HGDFW} for each $\alpha\!:\!(\preceq\!n\,r.C) \in \Gamma$ do
	\begin{enumerate}
	\item for each $\tuple{X,Y,\alpha} \in \mE$ do
	   \begin{enumerate}
	   \item if $r \in X$ and $\{C,\ovl{C}\} \cap Y = \emptyset$ then\\
		$\mE' := \mE' \cup \{\tuple{X,Y \cup \{C\},\alpha},\tuple{X,Y\cup\{\ovl{C}\},\alpha}\}$\\
	(i.e., $\tuple{X,Y,\alpha}$ is replaced by $\tuple{X,Y \cup \{C\},\alpha}$ and $\tuple{X,Y\cup\{\ovl{C}\},\alpha}$)
	   \item else $\mE' := \mE' \cup \{\tuple{X,Y,\alpha}\}$
	   \end{enumerate}
	\item $\mE := \mE'$, $\mE' := \emptyset$
	\end{enumerate}

     \item\label{item: HMFHS} repeat
	\begin{enumerate}
	\item[] for each $\alpha\!:\!(\preceq\!n\,r.C) \in \Gamma$, $\tuple{X,Y,\alpha} \in \mE$ and $\tuple{X',Y',\alpha} \in \mE$ such that $r \in X$, $C \in Y$, $r \in X'$, $C \in Y'$, $\tuple{X \cup X', Y \cup Y', \alpha} \notin \mE$ and $Y \cup Y'$ does not contain any pair of the form $\varphi$, $\ovl{\varphi}$ do add $\tuple{X \cup X', Y \cup Y', \alpha}$ to $\mE$ (i.e., the merger of $\tuple{X,Y,\alpha}$ and $\tuple{X',Y',\alpha}$ is added to $\mE$)
	\end{enumerate}
	until no tuples were added to $\mE$ during the last iteration

     \item\label{item: JHFSA} for each $\tuple{X,Y,\alpha} \in \mE$ do
	\begin{enumerate}
	\item[] $\ConToSucc(v,\NonState,\Simple,Y,\emptyset,\Null,\tuple{\CQF,X,\alpha})$
	\end{enumerate}

     \item\label{item: JHDSA0} $\ILConstraints(v) := \{ x_{w,e} \geq 0 \mid$ $\tuple{v,w} \in E$, $e \in \ELabels(v,w)$ and\\ 
	\mbox{\hspace{14.2em}}$\piT(e) = \CQF\}$
     \item\label{item: JEROS} for each $(\alpha\!:\!C) \in \Gamma$ do
       \begin{enumerate}
       \item if $C$ is of the form $\succeq\!n\,r.D$ 
	then add to $\ILConstraints(v)$ the constraint $\sum \{ x_{w,e} \mid$ $\tuple{v,w} \in E$, $e \in \ELabels(v,w)$, $\piT(e) = \CQF$, $r \in \piR(e)$, $\piI(e) = \alpha$, $D \in \Label(w)\} \geq n$

       \item if $C$ is of the form $\preceq\!n\,r.D$ 
	then add to $\ILConstraints(v)$ the constraint $\sum \{ x_{w,e} \mid$ $\tuple{v,w} \in E$, $e \in \ELabels(v,w)$, $\piT(e) = \CQF$, $r \in \piR(e)$, $\piI(e) = \alpha$, $D \in \Label(w)\} \leq n$
	\end{enumerate}

     \item $\Status(v) := \FExpanded$.
     \end{enumerate}

  \begin{explanation}
  Let $\Gamma$ be the set computed at the step~\ref{step: UIRJS}. It consists of the requirements to be realized for~$v$. 
  To satisfy a requirement $\varphi = (\alpha\!:\,\succeq\!n\,r.C) \in \Gamma$ for $v$, one can first connect $v$ to a successor $w_\varphi$ using an edge label $e$ specified by the tuple $\tuple{X,Y,\alpha}$ computed at the step~\ref{item: HJFDU}, where $\piT(e) = \CQF$, $\piR(e) = X$, $\piI(e) = \alpha$ and $Y$ represents $\Label(w_\varphi)$, and then clone $w_\varphi$ to create $n$ successors for $v$ (or only record the intention somehow). The label of $w_\varphi$ contains only formulas necessary for realizing the requirement $\alpha\!:\!\E r.C$ and related ones of the form $\alpha\!:\!\V r'.C'$ in $\Gamma$. 
  To satisfy requirements of the form $\alpha\!:\preceq\!n'\,r'.C'$ for $v$, where $r \sqsubseteq_\mR r'$, we tend to use only copies of $w_\varphi$ extended with either $C'$ or $\ovl{C'}$ (for easy counting) as well as the mergers of such extended nodes. 
  So, we first start with the set $\mE$ constructed at the step~\ref{item: HJFDU}, which consists of tuples with information about successors to be created for~$v$. 
  We then modify $\mE$ by taking necessary extensions of the nodes (see the step~\ref{item: HGDFW}). After that, we continue modifying $\mE$ by adding to it also appropriate mergers of nodes and edges (see the step~\ref{item: HMFHS}). Successors for $v$ are created at the step~\ref{item: JHFSA}. The number of copies of a node $w$ that are intended to be used as successors of $v$ using an edge label $e$ is represented by a variable $x_{w,e}$ (we will not actually create such copies). The case when $w$ would be a named individual and cannot be cloned is dealt with by the rule \USc (see Explanation~\ref{exp: HGALC}). The set $\ILConstraints(v)$ consisting of appropriate constraints about such variables are set at the steps \ref{item: JHDSA0}-\ref{item: JEROS}. 
\koniec
  \end{explanation}
  \end{description}
} 


\subsection{Properties of \CSHOQ-Tableaux}

Define the size of a knowledge base $\KB = \tuple{\mR,\mT,\mA}$ to be the number of bits used for the usual sequential representation of $\KB$. It is greater than the number of symbols occurring in $\KB$. If $\Size$ is the size of $\KB$ and $\leq\!n\,r.C$ or $\geq\!n\,r.C$ is a number restriction occurring in $\KB$ then:
\begin{itemize}
\item when numbers are coded in unary we have that $n \leq \Size$,  
\item when numbers are coded in binary we have that $n \leq 2^\Size$.  
\end{itemize}

\newcommand{\LemmaComplexity}{Let $\tuple{\mR,\mT,\mA}$ be a knowledge base in NNF of the logic \SHOQ and let $\Size$ be the size of $\tuple{\mR,\mT,\mA}$. Then a \CSHOQ-tableau for $\tuple{\mR,\mT,\mA}$ can be constructed in (at most) exponential time in~$\Size$  in the following cases:
\begin{enumerate}
\item numbers are coded in unary, 
\item numbers are coded in binary and $n \leq \Size$ for every concept $\leq\!n\,r.C$ occurring in $\tuple{\mR,\mT,\mA}$,  
\item numbers are coded in binary and $n \leq \Size$ for every concept $\geq\!n\,r.C$ occurring in $\tuple{\mR,\mT,\mA}$.
\myEnd
\end{enumerate}
} 
\begin{lemma}[Complexity]\label{lemma: Complexity}
\LemmaComplexity
\end{lemma}

\begin{theorem}[Soundness and Completeness]
\label{theorem: s-c}
Let $\tuple{\mR,\mT,\mA}$ be a knowledge base in NNF of the logic \SHOQ and $G = \tuple{V,E,\nu}$ be an arbitrary \CSHOQ-tableau for $\tuple{\mR,\mT,\mA}$. Then $\tuple{\mR,\mT,\mA}$ is satisfiable iff $\Status(\nu) \neq \Unsat$. 
\myEnd
\end{theorem}

\LongVersion{See the Appendix for the proofs of the above lemma and theorem.}
\ShortVersion{The proofs of the above lemma and theorem can be found in~\cite{SHOQ-long}.}

To check satisfiability of $\tuple{\mR,\mT,\mA}$ one can construct a \CSHOQ-tableau for it, then return ``no'' when the root of the tableau has status $\Unsat$, or ``yes'' in the other cases. We call this the {\em \CSHOQ-tableau decision procedure}. The corollary given below immediately follows from Theorem~\ref{theorem: s-c} and Lemma~\ref{lemma: Complexity}.

\begin{corollary}
The \CSHOQ-tableau decision procedure has \EXPTIME complexity when numbers are coded in unary.
\myEnd
\end{corollary}


\ShortVersion{\input{SHOQ-examples}}


\section{Conclusions}
\label{section: conc}

Recall that \SHIQ, \SHOQ, \SHIO are the three most well-known expressive DLs with \EXPTIME complexity. (Due to the interaction between $\mathcal{I}$, $\mathcal{Q}$ and $\mathcal{O}$, the complexity of the DL \SHOIQ is \NEXPTIME-complete).
\ShortVersion{In this paper, we have presented the first tableau method with an \EXPTIME (optimal) complexity for checking satisfiability of a knowledge base in the DL \SHOQ when numbers are coded in unary. Our detailed tableau decision procedure for \SHOQ is given in~\cite{SHOQ-long}.} 
\LongVersion{In this paper, we have presented the first \EXPTIME tableau decision procedure for checking satisfiability of a knowledge base in the DL \SHOQ when numbers are coded in unary.} 

We applied Nguyen's method~\cite{SHIQ-long} of integer linear feasibility checking for dealing with number restrictions. 
This work differs from the work~\cite{SHIQ-long} in that nominals are allowed instead of inverse roles. Without inverse roles, global caching is used instead of global state caching to allow more cache hits. We used special techniques for dealing with nominals and their interaction with number restrictions.

%% file: SHOQ-examples.tex
\LongVersion{\subsection{Illustrative Examples}

Before specifying the tableau rules in detail, we present simple examples to illustrate some ideas (but not all aspects) of our method. Despite that these examples refer to the tableau rules, we choose this place for presenting them because the examples are quite intuitive and the reader can catch the ideas of our method without knowing the detailed rules. He or she can always consult the rules in the next subsection.
}

\ShortVersion{\section{Illustrative Examples}}


\label{section: examples}

\begin{figure}[t!]
\begin{center}
\includegraphics{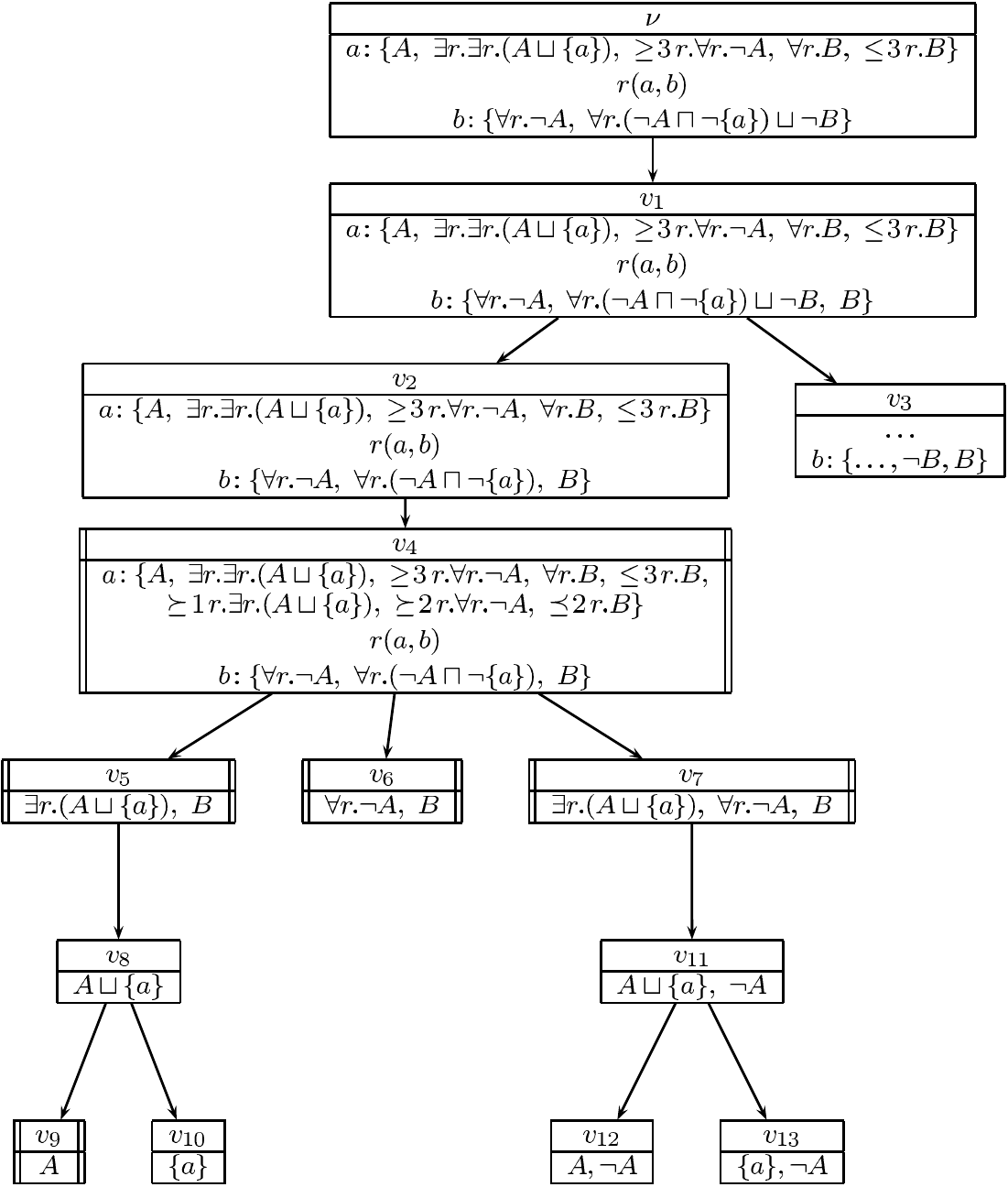}
\end{center}
\caption{An illustration of the tableau described in Example~\ref{example1}. 
The marked nodes $v_4$ -- $v_7$ and $v_9$ are states. The nodes $\nu$, $v_1$ -- $v_4$ are complex nodes, the remaining are simple nodes. In each node, we display the formulas of its label.
\label{fig-1}}
\end{figure} 

\begin{example}\label{example1}
Let us construct a \CSHOQ-tableau for $\tuple{\mR,\mT,\mA}$, where
\[
\begin{array}{rcl}
\mA & = & \{a\!:\!A,\ a\!:\!\E r.\E r.(A \mor \{a\}),\ a\!:\,\geq\!3\,r.\V r.\lnot A,\ a\!:\!\forall r.B,\ a\!:\,\leq\!3\,r.B,\\[0.5ex] 
& & \;\,r(a,b),\ b\!:\!\V r.\lnot A,\ b\!:\!(\V r.(\lnot A \mand \lnot\{a\}) \mor \lnot B)\},
\end{array}
\]
$\mR = \emptyset$ and $\mT = \emptyset$. 
An illustration is presented in Figure~\ref{fig-1}.

At the beginning, the graph has only the root $\nu$ which is a complex non-state with \mbox{$\Label(\nu) = \mA$}. 
Since $\{a\!:\!\V r.B, r(a,b)\} \subset \Label(\nu)$, applying a unary static expansion rule to $\nu$, we connect it to a new complex non-state $v_1$ with $\Label(v_1) = \Label(\nu) \cup \{b\!:\!B\}$. 

Since $b\!:\!(\V r.(\lnot A \mand \lnot\{a\}) \mor \lnot B) \in \Label(v_1)$, applying the non-unary static expansion rule to $v_1$, we connect it to new complex non-states $v_2$ and $v_3$ with 
\begin{eqnarray*}
\Label(v_2) & = & \Label(v_1) - \{b\!:\!(\V r.(\lnot A \mand \lnot\{a\}) \mor \lnot B)\} \cup \{b\!:\!\V r.(\lnot A \mand \lnot\{a\})\}\\
\Label(v_3) & = & \Label(v_1) - \{b\!:\!(\V r.(\lnot A \mand \lnot\{a\}) \mor \lnot B)\} \cup \{b\!:\!\lnot B\}. 
\end{eqnarray*}

Since both $b\!:\!B$ and $b\!:\!\lnot B$ belong to $\Label(v_3)$, the node $v_3$ receives the status $\Unsat$. 
Applying the forming-state rule to $v_2$, we connect it to a new complex state $v_4$ with 
\[ \Label(v_4) = \Label(v_2) \cup \{a\!:\,\succeq\!1\,r.\E r.(A \mor \{a\}),\ a\!:\,\succeq\!2\,r.\V r.\lnot A,\ a\!:\,\preceq\!2\,r.B\}. \]
The assertion $a\!:\,\succeq\!1\,r.\E r.(A \mor \{a\}) \in \Label(v_4)$ is due to $a\!:\!\E r.\E r.(A \mor \{a\}) \in \Label(v_2)$ and the fact that the negation of $b\!:\!\E r.(A \mor \{a\})$ in NNF belongs to $\Label(v_2)$ (notice that $r(a,b) \in \Label(v_2)$). 
The assertion $a\!:\,\succeq\!2\,r.\V r.\lnot A \in \Label(v_4)$ is due to $a\!:\,\geq\!3\,r.\V r.\lnot A \in \Label(v_2)$ and the fact that $\{r(a,b)$, \mbox{$b\!:\!\V r.\lnot A\} \subset \Label(v_2)$}. 
Similarly, the assertion $a\!:\,\preceq\!2\,r.B \in \Label(v_4)$ is due to $a\!:\,\leq\!3\,r.B \in \Label(v_2)$ and the fact $\{r(a,b), b\!:\!B\} \subset \Label(v_2)$.

As $r$ is a numeric role, applying the transitional partial-expansion rule\footnote{which is used for making transitions via non-numeric roles} to $v_4$, we just change the status of $v_4$ to $\PExpanded$. After that, applying the transitional full-expansion rule to $v_4$, we connect it to new simple non-states $v_5$, $v_6$, $v_7$, using the edge label $e = \tuple{\CQF,\{r\},a}$, such that $\Label(v_5) = \{\E r.(A \mor \{a\}),\ B\}$, $\Label(v_6) = \{\V r.\lnot A,\ B\}$, $\Label(v_7) = \{\E r.(A \mor \{a\}),\ \V r.\lnot A,\ B\}$. 
The creation of $v_5$ is caused by $a\!:\,\succeq\!1\,r.\E r.(A \mor \{a\}) \in \Label(v_4)$, while the creation of $v_6$ is caused by $a\!:\,\succeq\!1\,r.\V r.\lnot A$. The node $v_7$ results from merging $v_5$ and $v_6$. Furthermore, $\ILConstraints(v_4)$ consists of $x_{v_i,e} \geq 0$, for $5 \leq i \leq 7$, and 
\begin{eqnarray*}
x_{v_5,e} + x_{v_7,e} & \geq & 1\\
x_{v_6,e} + x_{v_7,e} & \geq & 2\\
x_{v_5,e} + x_{v_6,e} + x_{v_7,e} & \leq & 2.
\end{eqnarray*}

Applying the forming-state rule to $v_5$, the type of this node is changed from $\NonState$ to $\State$. Next, applying the transitional partial-expansion rule to $v_5$, its status is changed to $\PExpanded$. Then, applying the transitional full-expansion rule to $v_5$, we connect $v_5$ to a new simple non-state $v_8$ with $\Label(v_8) = \{A \mor \{a\}\}$ using the edge label $e' = \tuple{\CQF,\{r\},\Null}$ and set $\ILConstraints(v_5) = \{x_{v_8,e'} \geq 0, x_{v_8,e'} \geq 1\}$. 

Applying the non-unary static expansion rule to $v_8$, we connect it to new simple non-states $v_9$ and $v_{10}$ with $\Label(v_9) = \{A\}$ and $\Label(v_{10}) = \{\{a\}\}$. The status of $v_9$ is then changed to $\Sat$, which causes the statuses of $v_8$ and $v_5$ to be updated to $\Sat$. The node $v_{10}$ is not expanded as it does not affect the status of the root node $\nu$. 

Applying the forming-state rule to $v_6$, the type of this node is changed from $\NonState$ to $\State$. Next, applying the transitional partial-expansion rule and then the transitional full-expansion rule to $v_6$, its status is changed to $\FExpanded$. The status of $v_6$ is then updated to $\Sat$. 

Applying the forming-state rule to $v_7$, the type of this node is changed from $\NonState$ to $\State$. Next, applying the transitional partial-expansion rule to $v_7$, its status is changed to $\PExpanded$. Then, applying the transitional full-expansion rule to $v_7$, we connect $v_7$ to a new simple non-state $v_{11}$ with $\Label(v_{11}) = \{A \mor \{a\},\ \lnot A\}$ using the mentioned edge label $e'$ and set $\ILConstraints(v_7) =$ \mbox{$\{x_{v_{11},e'} \geq 0$}, \mbox{$x_{v_{11},e'} \geq 1\}$}. 

Applying the non-unary static expansion rule to $v_{11}$, we connect it to new simple non-states $v_{12}$ and $v_{13}$ with $\Label(v_{12}) = \{A, \lnot A\}$ and $\Label(v_{13}) = \{\{a\}, \lnot A\}$. The status of $v_{12}$ is then changed to $\Unsat$.
Since $a\!:\!A \in \Label(v_4)$, the status of $v_{13}$ is updated to $\UnsatWrt(\{v_4\})$, which causes the status of $v_{11}$ to be updated also to $\UnsatWrt(\{v_4\})$. As the set $\ILConstraints(v_7) \cup \{x_{v_{11},e'} = 0\}$ is infeasible, the status of $v_7$ is updated to $\UnsatWrt(\{v_4\})$.  
Next, as the set $\ILConstraints(v_4) \cup \{x_{v_7,e} = 0\}$ is infeasible, the status of $v_4$ is first updated to $\UnsatWrt(\{v_4\})$ and then to $\Unsat$. After that, the statuses of $v_2$, $v_1$, $\nu$ are sequentially updated to $\Unsat$. Thus, we conclude that the knowledge base $\tuple{\mR,\mT,\mA}$ is unsatisfiable. 
\myEnd
\end{example}


\begin{figure}[t!]
\begin{center}
\includegraphics{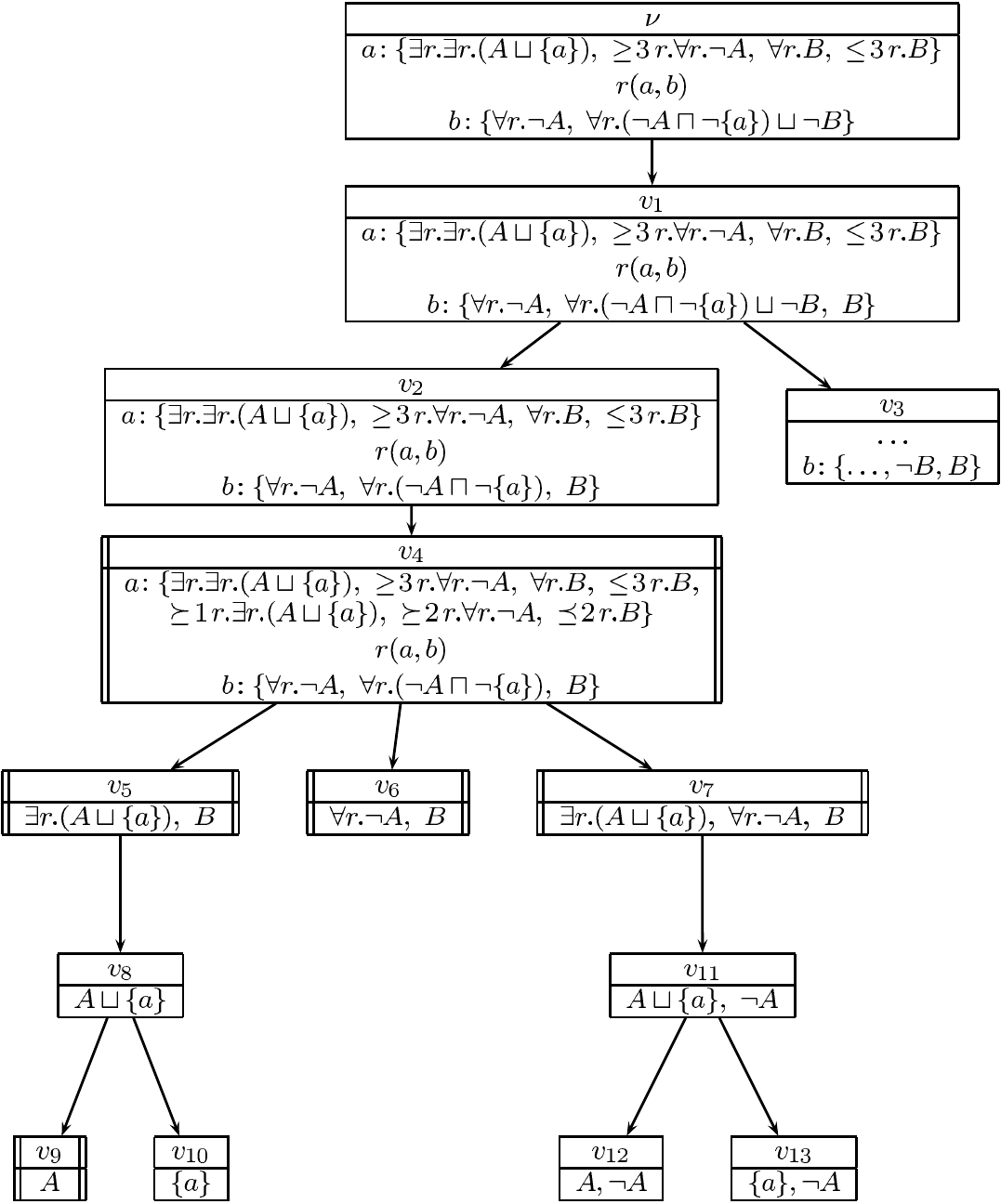}
\end{center}
\caption{An illustration for Example~\ref{example2} -- Part~I.\label{fig-2}}
\end{figure} 


\begin{figure}[t!]
\begin{center}
\includegraphics[scale=1.0]{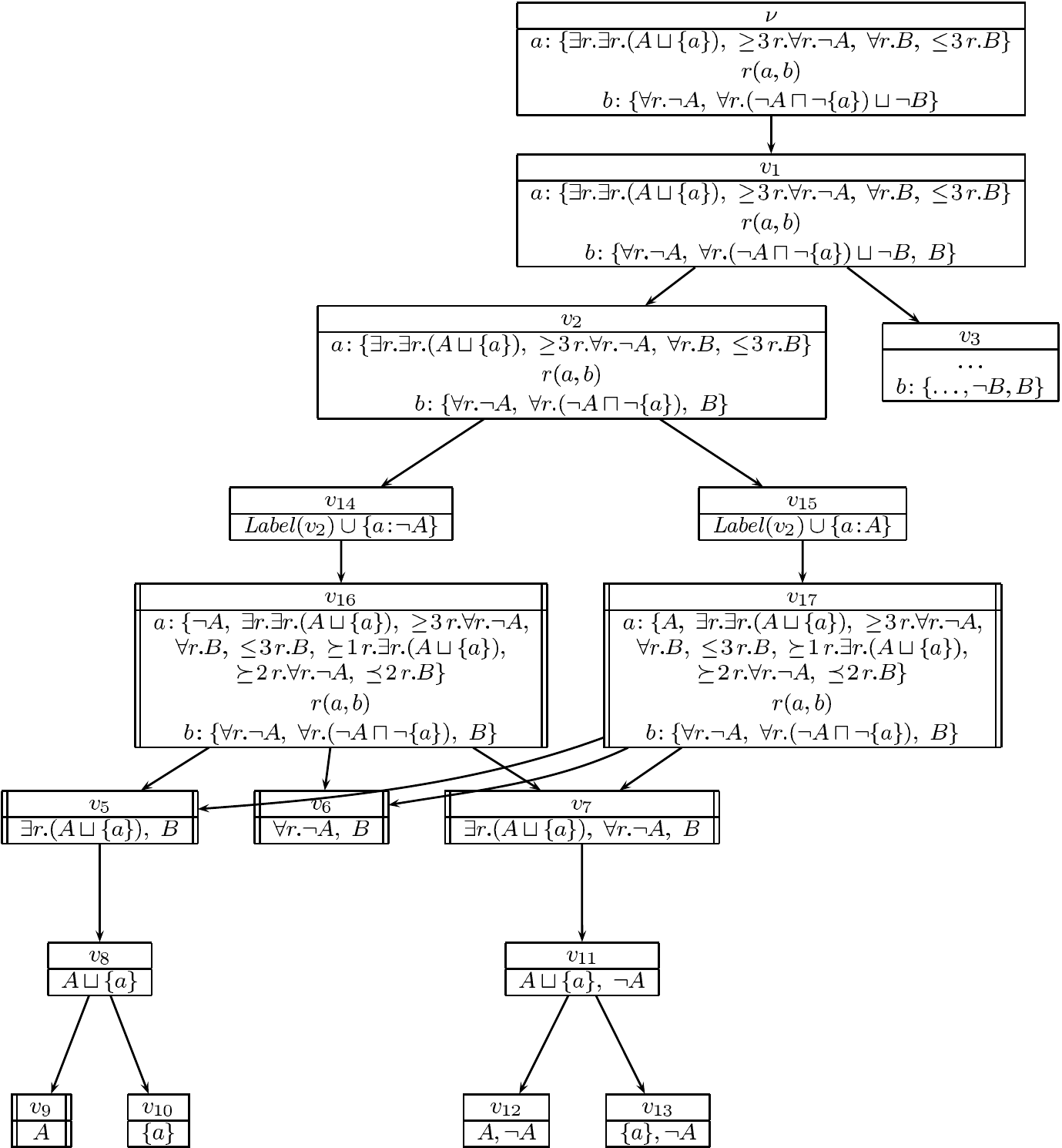}
\end{center}
\caption{An illustration for Example~\ref{example2} -- Part~II.\label{fig-3}}
\end{figure} 


\begin{example}\label{example2}
Let us modify Example~\ref{example1} by deleting the assertion $a\!:\!A$ from the ABox. That is, we are now constructing a \CSHOQ-tableau for $\tuple{\mR,\mT,\mA}$, where
\[
\begin{array}{rcl}
\mA & = & \{a\!:\!\E r.\E r.(A \mor \{a\}),\ a\!:\,\geq\!3\,r.\V r.\lnot A,\ a\!:\!\forall r.B,\ a\!:\,\leq\!3\,r.B,\\[0.5ex] 
& & \;\,r(a,b),\ b\!:\!\V r.\lnot A,\ b\!:\!(\V r.(\lnot A \mand \lnot\{a\}) \mor \lnot B)\},
\end{array}
\]
$\mR = \emptyset$ and $\mT = \emptyset$. 
The first stage of the construction is similar to the one of Example~\ref{example1}, up to the step of updating the status of $v_{12}$ to $\Unsat$. This stage is illustrated in Figure~\ref{fig-2}, which is similar to Figure~\ref{fig-1} except that the labels of the nodes $\nu$ and $v_1$ -- $v_4$ do not contain $a\!:\!A$. The continuation is described below and illustrated by Figure~\ref{fig-3}. 

Since $\Label(v_{13}) = \{\{a\}, \lnot A\}$, applying the rule for dealing with nominals to $v_{13}$, we delete the edge $\tuple{v_2,v_4}$ (from $E$) and re-expand $v_2$ by connecting it to new complex non-states $v_{14}$ and $v_{15}$ with $\Label(v_{14}) = \Label(v_2) \cup \{a\!:\!\lnot A\}$ and $\Label(v_{15}) = \Label(v_2) \cup \{a\!:\!A\}$ as shown in Figure~\ref{fig-3}. The status of $v_{13}$ is updated to $\Blocked$. The node $v_4$ is not deleted, but we do not display it in Figure~\ref{fig-3}. 

Applying the forming-state rule to $v_{14}$ we connect it to a new complex state $v_{16}$. The label of $v_{16}$ is computed using $\Label(v_{14})$ in a similar way as in Example~\ref{example1} when computing $\Label(v_4)$. 

Applying the transitional partial-expansion rule to $v_{16}$ we change its status to $\PExpanded$. After that, applying the transitional full-expansion rule to $v_{16}$ we connect it to the existing nodes $v_5$, $v_6$, $v_7$ using the edge label $e = \tuple{\CQF,\{r\},a}$. The set $\ILConstraints(v_{16})$ is the same as $\ILConstraints(v_4)$.  

Applying the forming-state rule to $v_{15}$ we connect it to a new complex state $v_{17}$. The label of $v_{17}$ is computed using $\Label(v_{15})$ in a similar way as in Example~\ref{example1} when computing $\Label(v_4)$. 

The expansion of $v_{17}$ is similar to the expansion of $v_{16}$. The set $\ILConstraints(v_{17})$ is the same as $\ILConstraints(v_{16})$ and $\ILConstraints(v_4)$. Analogously to updating the statuses of the nodes $v_{13}$, $v_{11}$, $v_7$ in Example~\ref{example1} to $\UnsatWrt(\{v_4\})$, the statuses of $v_{13}$, $v_{11}$, $v_7$ are updated to $\UnsatWrt(\{v_{17}\})$. Next, as $\ILConstraints(v_{17}) \cup \{x_{v_7,e} = 0\}$ is infeasible, the status of $v_{17}$ is first updated to $\UnsatWrt(\{v_{17}\})$ and then to $\Unsat$. After that, the status of $v_{15}$ is also updated to $\Unsat$. 
As no more changes that may affect the status of $\nu$ can be made and $\Status(\nu) \neq \Unsat$, we conclude that the knowledge base $\tuple{\mR,\mT,\mA}$ is satisfiable.
\myEnd
\end{example}

%% file: SHOQ-proofs.tex
\newpage
\appendix

\section{Appendix: Proofs}
\label{section: proofs}

\subsection{Complexity}

Let $\Size$ be the size of $\tuple{\mR,\mT,\mA}$. 
Define $\closure(\mR,\mT,\mA)$ to be the smallest set $\Gamma$ of formulas such that:
\begin{enumerate}
\item\label{item JHFDO 1} all concepts (and subconcepts) used in $\tuple{\mR,\mT,\mA}$ belong to $\Gamma$,
\item\label{item JHFDO 1b} if $r$ is a role and $a$ is an individual used in $\tuple{\mR,\mT,\mA}$ then $(\leq\!1\,r.\{a\}) \in \Gamma$, 
\item\label{item JHFDO 2} if $\V s.C \in \Gamma$ and $r \sqsubseteq_\mR s$ then $\V r.C \in \Gamma$,
\item\label{item JHFDO 3} if $\leq\!0\,s.C \in \Gamma$ then $\V s.\ovl{C} \in \Gamma$,
\item\label{item JHFDO 4} if $C \in \Gamma$ and $C$ is not of the form $\preceq\!n\,r.C$ nor $\succeq\!n\,r.C$ then $\ovl{C} \in \Gamma$,  
\item\label{item JHFDO 5} if $\E r.C \in \Gamma$ and $r$ is a numeric role then $\succeq\!1\,r.C \in \Gamma$,  
\item\label{item JHFDO 6} if $\geq\!n\,r.C \in \Gamma$, $0 \leq m \leq \Size$ and $m < n$ then $\succeq\!(n-m)\,r.C \in \Gamma$,  
\item\label{item JHFDO 7} if $\leq\!n\,r.C \in \Gamma$, $0 \leq m \leq \Size$ and $m \leq n$ then $\preceq\!(n-m)\,r.C \in \Gamma$,  
\item\label{item JHFDO 8} all assertions of $\mA$ belong to $\Gamma$, 
\item\label{item JHFDO 9} if $C \in \Gamma$ and $a$ is an individual used in $\tuple{\mR,\mT,\mA}$ then $a\!:\!C \in \Gamma$, 
\item\label{item JHFDO 10} if $a$ and $b$ are individuals used in $\tuple{\mR,\mT,\mA}$ then $a \doteq b$ and $a \not\doteq b$ belong to $\Gamma$, 
\item\label{item JHFDO 11} if $r$ is a role and $a$, $b$ are individuals used in $\tuple{\mR,\mT,\mA}$ then $r(a,b)$ and $\lnot r(a,b)$ belong to~$\Gamma$. 
\end{enumerate}

\begin{lemma}
The number of formulas of $\closure(\mR,\mT,\mA)$ is of rank $O(\Size^3)$, where $\Size$ is the size of $\tuple{\mR,\mT,\mA}$. 
\end{lemma}

\begin{proof}
The set $\Gamma = \closure(\mR,\mT,\mA)$ can be constructed by initializing $\Gamma$ according to the items~\ref{item JHFDO 1} and \ref{item JHFDO 8}, and then repeatedly applying the rules stated in the remaining items of the list. After initialization, the set $\Gamma$ has $O(\Size)$ formulas. The rules in the items~\ref{item JHFDO 1b}-\ref{item JHFDO 7} add $O(\Size^2)$ formulas to~$\Gamma$.  
The rule in the item~\ref{item JHFDO 9} adds $O(\Size^3)$ formulas to $\Gamma$ (as $\Gamma$ contains $O(\Size^2)$ concepts and there are $O(\Size)$ individual names). The rules in the items~\ref{item JHFDO 10} and~\ref{item JHFDO 11} add $O(\Size^3)$ formulas to $\Gamma$. Thus, at the end, $\Gamma$ is of rank $O(\Size^3)$. 
\myEnd
\end{proof}

We recall below Lemma~\ref{lemma: Complexity} before presenting its proof. 

\medskip

\noindent\textbf{Lemma~\ref{lemma: Complexity}.} {\em \LemmaComplexity}
\begin{proof}
Let us construct an arbitrary \CSHOQ-tableau $G = \tuple{V,E,\nu}$ for $\tuple{\mR,\mT,\mA}$. 

Let $\Size'$ be the number of formulas of $\closure(\mR,\mT,\mA)$. We have $\Size' = O(\Size^3)$. 
For each $v \in V$, $\Label(v) \subseteq \closure(\mR,\mT,\mA)$. Since nodes of $G$ are globally cached, it follows that $G$ has no more than $2^{\Size'}$ nodes. 

Each state has no more than $O(2^\Size \cdot 2^{\Size'} \cdot \Size)$ outgoing edges (since each outgoing edge created by the transitional full-expansion rule is characterized by a tuple $\tuple{X,Y,\alpha}$, where $X$ is a set of roles, $Y$ is a set of concepts, and $\alpha$ is $\Null$ or an individual name).\footnote{The bound can be made tighter, e.g., using $O(\Size^2)$ instead of $\Size'$.} Thus, checking feasibility of $\ILConstraints(v)$ for a state $v$ is an $\IFDL{\Size$, $2^\Size \cdot 2^{\Size'} \cdot \Size$, $\Size}$-problem that satisfies the assumptions of Lemma~\ref{lemma: IFDL} for the first case and satisfies the assumptions of Lemma~\ref{lemma: IFDL2} for the remaining two cases, and hence can be solved in (at most) exponential time in~$\Size$. 

Therefore, choosing a node to expand, checking whether a rule is applicable, and applying a rule can be done in (at most) exponential time in~$\Size$. As each node is re-expanded at most once (by the rule \DN), we conclude that the graph $G$ can be constructed in (at most) exponential time in~$\Size$ for the considered cases.
\myEnd
\end{proof}

\subsection{Soundness}

\begin{lemma} \label{lemma: UYDHW}
Let $G = \tuple{V,E,\nu}$ be a \CSHOQ-tableau for $\tuple{\mR,\mT,\mA}$. Then, for every $v \in V$, $\FullLabel(v)$ is equivalent to $\Label(v)$. That is, for any interpretation $\mI$: 
\begin{itemize}
\item if $v$ is a simple node then $(\FullLabel(v))^\mI = (\Label(v))^\mI$,
\item if $v$ is a complex node then $\mI$ satisfies $\FullLabel(v)$ iff it satisfies $\Label(v)$.
\end{itemize}
\end{lemma}

The proof of this lemma is straightforward. 

\begin{lemma} \label{lemma: IUFDN}
Let $G = \tuple{V,E,\nu}$ be a \CSHOQ-tableau for $\tuple{\mR,\mT,\mA}$. Then, for every $v \in V$, if $\Type(v) = \NonState$ and $w_1,\ldots,w_k$ are all the successors of $v$ then $\FullLabel(v)$ is satisfiable w.r.t.\ $\mR$ and $\mT$ iff there exists $1 \leq i \leq k$ such that $\FullLabel(w_i)$ is satisfiable w.r.t.~$\mR$ and~$\mT$.
\end{lemma}

The proof of this lemma is straightforward. 

Let $G = \tuple{V,E,\nu}$ be a \CSHOQ-tableau. For each node $v$ of $G$ with $\Status(v) \in \{\Unsat$, $\Sat\}$, let $\DSTimeStamp(v)$ be the moment at which $\Status(v)$ was changed to its final value (i.e., determined to be $\Unsat$ or $\Sat$). $\DSTimeStamp$ stands for ``determined-status time-stamp''. 
For each node $v$ of $G$ with $\Status(v) = \UnsatWrt(U)$ and for each $u \in U$, let $\CSTimeStamp(v,u)$ be the first moment at which $\Status(v)$ was changed to some $\UnsatWrt(U')$ with $u \in U'$. $\CSTimeStamp$ stands for ``$\UnsatWrt$-status time-stamp''. 
For each non-state $v$ of $G$, let $\ETimeStamp(v)$ be the moment at which $v$ was expanded the last time.\footnote{Each non-state may be re-expanded at most once (by the rule \DN) and each state is expanded at most once.}

\begin{lemma} \label{lemma: SHQWD}
Let $G = \tuple{V,E,\nu}$ be a \CSHOQ-tableau for $\tuple{\mR,\mT,\mA}$. Then, for every $v \in V :$
\begin{enumerate}
\item if $\Status(v) = \Unsat$ then $\FullLabel(v)$ is unsatisfiable w.r.t.\ $\mR$ and $\mT$, 
\item if $\Status(v) = \UnsatWrt(U)$ and $u \in U$ then there does not exist any model of $\tuple{\mR,\mT,\FullLabel(u)}$ that satisfies $\FullLabel(v)$.
\end{enumerate}
\end{lemma}

\begin{proof} 
We prove this lemma by induction on the above mentioned time-stamps. 

Consider the first assertion of the lemma for the case when $v$ is a simple state and $\Status(v)$ is changed to $\Unsat$ by the subrule~\ref{item: HGWSS} of \UPSc because $\ILConstraints(v)$ is infeasible. We prove the contrapositive: suppose $\mI$ is a model of $\mR$ and $\mT$, and $y \in (\FullLabel(v))^\mI$; we show that $\ILConstraints(v)$ is feasible.
Without loss of generality, assume that $\mI$ is finitely-branching.\footnote{It is known that the DL \SHOQ has the finitely-branching model property.} 
Thus, the set $Z = \{z \in \Delta^\mI \mid \tuple{y,z} \in r^\mI$ for some $r \in \RN\}$ is finite. 
We compute a solution $\mS$ for $\ILConstraints(v)$ as follows.

\begin{itemize}
\item For each $\tuple{v,w} \in E$ and $e \in \ELabels(v,w)$ such that $\piT(e) = \CQF$, set \mbox{$n_{w,e} := 0$}.
\item For each $z \in Z$ do:
  \begin{itemize}
  \item let $\tuple{w_1,e_1},\ldots,\tuple{w_k,e_k}$ be all the pairs such that, for each $1 \leq i \leq k\,$:
     \begin{itemize}
     \item $\tuple{v,w_i} \in E$, $e_i \in \ELabels(v,w_i)$ and $\piT(e_i) = \CQF$, 
     \item $z \in (\FullLabel(w_i))^\mI$,
     \item $\tuple{y,z} \in r^\mI$ for all $r \in \piR(e_i)$,
     \item the pair $\tuple{w_i,e_i}$ is ``maximal'' in the sense that there does not exist any pair $\tuple{w'_i,e'_i} \neq \tuple{w_i,e_i}$ such that
	\begin{itemize}
	\item $\tuple{v,w'_i} \in E$, $e'_i \in \ELabels(v,w'_i)$ and $\piT(e'_i) = \CQF$, 
	\item $z \in (\FullLabel(w'_i))^\mI$,
	\item $\tuple{y,z} \in r^\mI$ for all $r \in \piR(e'_i)$;
	\end{itemize}
     \end{itemize} 

  \item for each $1 \leq i \leq k$, set $n_{w_i,e_i} := n_{w_i,e_i} + 1$.
  \end{itemize}
\item $\mS := \{x_{w,e} = n_{w,e} \mid \tuple{v,w} \in E$, $e \in \ELabels(v,w)$ and $\piT(e) = \CQF\}$.
\end{itemize}

We prove that $\mS$ is a solution for $\ILConstraints(v)$. 

If a constraint $x_{w,e} = 0$ was added to $\ILConstraints(v)$ because $w$ got status $\Unsat$ then, by the inductive assumption with $v$ replaced by $w$, we can conclude that $n_{w,e}$ was not increased at all and hence must be 0, which means that the constraint $x_{w,e} = 0$ is satisfied by the solution~$\mS$. 

Consider a constraint $\sum \{ x_{w,e} \mid$ $\tuple{v,w} \in E$, $e \in \ELabels(v,w)$, $\piT(e) = \CQF$, $r \in \piR(e)$, $\piI(e) = \Null$, $D \in \Label(w)\} \geq n$ of $\ILConstraints(v)$ and the corresponding concept $\succeq\!n\,r.D$. By the assumptions about $v$ and $y$, it can be derived that $Z$ contains pairwise different $z_1, \ldots, z_n$ such that $\tuple{y,z_i} \in r^\mI$ and $z_i \in D^\mI$, for $1 \leq i \leq n$. Each $z_i$ makes $n_{w,e}$ increased by 1 for some pair $\tuple{w,e}$ such that $\tuple{v,w} \in E$, $e \in \ELabels(v,w)$, $\piT(e) = \CQF$, $r \in \piR(e)$ and $D \in \FullLabel(w)$. Therefore, the considered constraint is satisfied by the solution~$\mS$.

Consider a constraint $\sum \{ x_{w,e} \mid$ $\tuple{v,w} \in E$, $e \in \ELabels(v,w)$, $\piT(e) = \CQF$, $r \in \piR(e)$, $\piI(e) = \Null$, $D \in \Label(w)\} \leq n$ of $\ILConstraints(v)$ and the corresponding concept $\preceq\!n\,r.D$. By the assumptions about $v$ and $y$, it can be derived that $Z$ contains no more than $n$ pairwise different elements $z_1, \ldots, z_n$ such that $\tuple{y,z_i} \in r^\mI$ and $z_i \in D^\mI$, for $1 \leq i \leq n$. For each $z_i$, there exists at most one pair $\tuple{w,e}$ such that $\tuple{v,w} \in E$, $e \in \ELabels(v,w)$, $\piT(e) = \CQF$, $r \in \piR(e)$, $D \in \FullLabel(w)$ and the consideration of $z_i$ causes $n_{w,e}$ to be increased by 1. This is due to the ``maximality'' of $\tuple{w,e}$ and the nature of the transitional full-expansion rule. Therefore, the considered constraint is satisfied by the solution~$\mS$.

Now, consider the second assertion of the lemma for the case when $v$ is a complex state and $\Status(v)$ becomes $\UnsatWrt(U)$ with $u \in U$ because of the call of $\SetClosedWrt(v,u)$ at the step~\ref{item:IRDMS} of the rule \UPSc due to infeasibility of the set $\ILCh = \ILConstraints(v)\ \cup$ $\{x_{w_i,e_i} = 0 \mid$ $1 \leq i \leq k\}$, where $\tuple{w_1,e_1}$, \ldots, $\tuple{w_k,e_k}$ are the pairs mentioned at that step of \UPSc. 
We prove the contrapositive: suppose $\mI$ is a model of $\tuple{\mR,\mT,\FullLabel(u)}$ that satisfies $\FullLabel(v)$; we show that the mentioned set $\ILCh$ of constraints is feasible. Without loss of generality, assume that $\mI$ is finitely-branching. 
We compute a solution $\mS$ for $\ILCh$ as follows.
 
\begin{enumerate}
\item For each $\tuple{v,w} \in E$ and $e \in \ELabels(v,w)$ such that $\piT(e) = \CQF$, set \mbox{$n_{w,e} := 0$}.
\item For each individual $a$ occurring in $\Label(v)$ and each $z \in \Delta^\mI$ such that $\tuple{a^\mI,z} \in r^\mI$ for some $r \in \RN$ do:
  \begin{enumerate}
  \item let $\tuple{w'_1,e'_1},\ldots,\tuple{w'_{k'},e'_{k'}}$ be all the pairs such that, for each $1 \leq i \leq k'\,$:
     \begin{enumerate}
     \item $\tuple{v,w'_i} \in E$, $e'_i \in \ELabels(v,w'_i)$, $\piT(e'_i) = \CQF$ and $\piI(e'_i) = a$, 
     \item $z \in (\FullLabel(w'_i))^\mI$,
     \item $\tuple{a^\mI,z} \in r^\mI$ for all $r \in \piR(e'_i)$,
     \item the pair $\tuple{w'_i,e'_i}$ is ``maximal'' in the sense that there does not exist any pair $\tuple{w''_i,e''_i} \neq \tuple{w'_i,e'_i}$ such that
	\begin{itemize}
	\item $\tuple{v,w''_i} \in E$, $e''_i \in \ELabels(v,w''_i)$, $\piT(e''_i) = \CQF$ and $\piI(e''_i) = a$, 
	\item $z \in (\FullLabel(w''_i))^\mI$,
	\item $\tuple{a^\mI,z} \in r^\mI$ for all $r \in \piR(e''_i)$;
	\end{itemize}
     \end{enumerate} 

  \item for each $1 \leq i \leq k'$ do
     \begin{enumerate}
     \item\label{item: XRIYH} if $z \neq b^\mI$ for all $b$ occurring in $\Label(v)$ then $n_{w'_i,e'_i} := n_{w'_i,e'_i} + 1$;
     \item\label{item: XRIXH} else if $z = b^\mI$ for some $b$ occurring in $\Label(v)$ and there exists $a\!:\!(\succeq\!l\,s.D) \in \Label(v)$ such that $s \in \piR(e'_i)$, $D \in \FullLabel(w'_i)$ and $s(a,b) \notin \Label(v)$ then $n_{w'_i,e'_i} := n_{w'_i,e'_i} + 1$.
     \end{enumerate}
  \end{enumerate}
\item $\mS := \{x_{w,e} = n_{w,e} \mid \tuple{v,w} \in E$, $e \in \ELabels(v,w)$ and $\piT(e) = \CQF\}$.
\end{enumerate}

We prove that $\mS$ is a solution for $\ILCh$. 

If a constraint $x_{w,e} = 0$ was added to $\ILConstraints(v)$ because $w$ got status $\Unsat$ then, by the inductive assumption with $v$ replaced by $w$, we can conclude that $n_{w,e}$ was not increased at all and hence must be 0, which means that the constraint $x_{w,e} = 0$ is satisfied by the solution~$\mS$. 

Consider a constraint $(x_{w_i,e_i} = 0) \in \ILCh$ with $1 \leq i \leq k$ (the pair $\tuple{w_i,e_i}$ was mentioned earlier). By the specification of $\tuple{w_i,e_i}$ (at the step~\ref{item:IRDMS} of the rule \UPSc), $\Status(w_i)$ is of the form $\UnsatWrt(U_i)$ with $u \in U_i$. By the inductive assumption with $v$ replaced by $w_i$, we can conclude that $n_{w_i,e_i}$ was not increased at all and hence must be 0, which means that the constraint $x_{w_i,e_i} = 0$ is satisfied by the solution~$\mS$. 

Consider a concept $(a\!:\,\succeq\!n\,s.D) \in \Label(v)$ and the corresponding constraint $\sum \{ x_{w,e} \mid$ $\tuple{v,w} \in E$, $e \in \ELabels(v,w)$, $\piT(e) = \CQF$, $s \in \piR(e)$, $\piI(e) = a$, \mbox{$D \in \Label(w)\} \geq n$} of $\ILConstraints(v)$. Let $m = \sharp\{b \mid \{s(a,b), b\!:\!D\} \subseteq \FullLabel(v)\}$. We have that either $(a\!:\,\geq\!(n+m)\,s.D) \in \Label(v)$ or $n=1$, $m=0$ and $(a\!:\E s.D) \in \Label(v)$. Since $\mI$ is a model of $\FullLabel(v)$, there exist pairwise different $z_1,\ldots,z_{n+m}$ such that $\tuple{a^\mI,z_i} \in s^\mI$ and $z_i \in D^\mI$ for all $1 \leq i \leq n+m$. Note that, if $s(a,b) \in \Label(v)$ then, by the subrule~\ref{item: HGW3A} of \NUS, either $b\!:\!D \in \FullLabel(v)$ or $b\!:\!\ovl{D} \in \FullLabel(v)$. Since $z_i \in D^\mI$ and $\mI$ is a model of $\FullLabel(v)$, if $z_i = b^\mI$ then $b\!:\!\ovl{D} \notin \FullLabel(v)$. Therefore, for every $1 \leq i \leq n+m$, if $z_i = b^\mI$ and $s(a,b) \in \Label(v)$ then $b\!:\!D \in \FullLabel(v)$. Let $Z = \{z_1,\ldots,z_{n+m}\} - \{b^\mI \mid s(a,b) \in \Label(v)\}$. We have that $\sharp Z = n$. Each $z$ from $Z$ makes $n_{w,e}$ increased by~1 for some pair $\tuple{w,e}$ such that $\tuple{v,w} \in E$, $e \in \ELabels(v,w)$, $\piT(e) = \CQF$, $\piI(e) = a$, $s \in \piR(e)$ and $D \in \Label(w)$. It follows that the considered constraint is satisfied by the solution~$\mS$.

Consider a concept $(a\!:\,\preceq\!n\,r.C) \in \Label(v)$ and the corresponding constraint $\sum \{ x_{w,e} \mid$ $\tuple{v,w} \in E$, $e \in \ELabels(v,w)$, $\piT(e) = \CQF$, $r \in \piR(e)$, $\piI(e) = a$, \mbox{$C \in \Label(w)\} \leq n$} of $\ILConstraints(v)$. Let $m = \sharp\{b \mid \{r(a,b), b\!:\!C\} \subseteq \FullLabel(v)\}$. We have that $(a\!:\,\leq\!(n+m)\,r.C) \in \Label(v)$. Since $\mI$ is a model of $\FullLabel(v)$, it follows that $a^\mI \in (\leq\!(n+m)\,r.C)^\mI$. Let $Z_1 = \{b^\mI \mid \{r(a,b), b\!:\!C\} \subseteq \FullLabel(v)\}$. Due to the subrule~\ref{item: JHEAA} of~\NUS, we have that $\sharp Z_1 = m$. 

Note that if $\tuple{v,w} \in E$, $e \in \ELabels(v,w)$, $\piT(e) = \CQF$, $\piI(e) = a$, $r \in \piR(e)$ and $C \in \FullLabel(w)$ then $n_{w,e}$ is increased only due to some $z$ such that $\tuple{a^\mI,z} \in r^\mI$ and $z \in C^\mI$. 
Due to the ``maximality'' of $\tuple{w,e}$ and the nature of the transitional full-expansion rule, for such a~$z$ there exists at most one pair $\tuple{v,w}$ such that $\tuple{v,w} \in E$, $e \in \ELabels(v,w)$, $\piT(e) = \CQF$, $\piI(e) = a$, $r \in \piR(e)$, $C \in \Label(w)$ and the consideration of $z$ causes $n_{w,e}$ to be increased by~1. 
Since $a^\mI \in (\leq\!(n+m)\,r.C)^\mI$, to prove that the considered constraint is satisfied by the solution~$\mS$, it suffices to show that if $z \in Z_1$ causes $n_{w'_i,e'_i}$ to be increased by~1 at the step~\ref{item: XRIXH} then $r \notin \piR(e'_i)$ or $C \notin \Label(w'_i)$. Suppose the contrary. We have that:
\begin{eqnarray}
& - & \{a\!:\,\leq\!(n+m)\,r.C, r(a,b), b\!:\!C\} \subseteq \FullLabel(v),\label{eq: ODJSA}\\
& - & \tuple{v,w'_i} \in E,\; e'_i \in \ELabels(v,w'_i),\; \piT(e'_i) = \CQF \mbox{ and } \piI(e'_i) = a,\label{eq: HJFGA}\\
& - & b^\mI \in (\Label(w'_i))^\mI \mbox{ and } \tuple{a^\mI,b^\mI} \in (r')^\mI \mbox{ for all } r' \in \piR(e'_i),\label{eq: IUFDA}\\
& - & a\!:\!(\succeq\!l\,s.D) \in \Label(v),\; s \in \piR(e'_i),\; D \in \Label(w'_i) \mbox{ and } s(a,b) \notin \Label(v),\label{eq: KLDFI}\\
& - & r \in \piR(e'_i) \mbox{ and } C \in \Label(w'_i).\label{eq: JDFLA} 
\end{eqnarray}

Since both $s$ and $r$ belong to $\piR(e'_i)$ (by~\eqref{eq: KLDFI} and~\eqref{eq: JDFLA}), there exist roles
\begin{equation}\label{eq: FDKJO}
\parbox{14cm}{$r_0 = r, r_1, \ldots, r_{h-1}, r_h = s$ and $s_1, \ldots, s_h$, all belonging to $\piR(e'_i)$}
\end{equation}
such that, for every $1 \leq j \leq h$:
\begin{eqnarray}
& - & s_j \sqsubseteq_\mR r_{j-1} \mbox{ and } s_j \sqsubseteq_\mR r_j,\label{eq: KJDFS}\\
& - & \mbox{$\Label(v)$ contains $a\!:\!\E s_j.D'_j$ or $a\!:\,\geq\!n_j\,s_j.D'_j$ for some $D'_j \in \Label(w'_i)$ and $n_j > 0$},\label{eq: JHFDS}\\
& - & \mbox{if $j < h$ then $\Label(v)$ contains $a\!:\,\leq\!m_j\,r_j.C'_j$ for some $C'_j \in \Label(w'_i)$ and $m_j$.}\label{eq: HJFGS}
\end{eqnarray}

Note that the subrule~\ref{item: OSJRS} of \NUS was not applicable to $v$. Having \eqref{eq: ODJSA}, \eqref{eq: FDKJO}, \eqref{eq: KJDFS} and \eqref{eq: JHFDS}, we derive that $s_1(a,b) \in \Label(v)$ or $\lnot s_1(a,b) \in \Label(v)$.  Since~\eqref{eq: FDKJO} and~\eqref{eq: IUFDA}, $\tuple{a^\mI,b^\mI} \in s_1^\mI$. Since $\mI$ is a model of $\FullLabel(v)$, it follows that $\lnot s_1(a,b) \notin \Label(v)$, and hence $s_1(a,b) \in \Label(v)$. Since $s_1 \sqsubseteq r_1$ (by~\eqref{eq: KJDFS}), by the rule \USa, we also have that $r_1(a,b) \in \Label(v)$. Analogously, using also~\eqref{eq: HJFGS}, for every $j$ from 1 to $h$, we can derive that $s_j(a,b) \in \Label(v)$ and $r_j(a,b) \in \Label(v)$. Since $s = r_h$, it follows that $s(a,b) \in \Label(v)$, which contradicts~\eqref{eq: KLDFI}. 
This completes the induction step for the second assertion of the lemma for the case when $v$ is a complex state and $\Status(v)$ becomes $\UnsatWrt(U)$ with $u \in U$ because of the call of $\SetClosedWrt(v,u)$ at the step~\ref{item:IRDMS} of the rule \UPSc due to infeasibility of $\ILCh$. 

The induction steps for:
\begin{itemize}
\item the first assertion of the lemma for the case when $v$ is a complex state and $\Status(v)$ is changed to $\Unsat$ by the subrule~\ref{item: HGWSS} of \UPSc because $\ILConstraints(v)$ is infeasible, 
\item the second assertion of the lemma for the case when $v$ is a simple state and $\Status(v)$ becomes $\UnsatWrt(U)$ with $u \in U$ because of the call of $\SetClosedWrt(v,u)$ at the step~\ref{item:IRDMS} of the rule \UPSc due to infeasibility of the corresponding set of constraints
\end{itemize}
can be proved in a similar way as done above for the two dual cases. 

The induction steps for the other cases (that correspond to the subrule~\ref{item: UFRMS} of \UPSa, the rule \UPSb and the subrules~\ref{item: YUSAO}, \ref{item: YUSLO}, \ref{item: YTDSA}, \ref{item: IDSMG} of \UPSc) are straightforward. 
\myEnd
\end{proof}

\begin{corollary}[Soundness of \CSHOQ]
If $G = \tuple{V,E,\nu}$ is a \CSHOQ-tableau for $\tuple{\mR,\mT,\mA}$ and $\Status(\nu) = \Unsat$ then $\tuple{\mR,\mT,\mA}$ is unsatisfiable.
\end{corollary}

This corollary directly follows from Lemma~\ref{lemma: SHQWD}.


\subsection{Completeness}

We prove completeness of \CSHOQ via model graphs. The technique has been used for other logics (e.g., in~\cite{Rautenberg83,Gore99,nguyen01B5SL,dkns2011}).
A {\em model graph} is a tuple $\langle \Delta, \nI, \nC, \nE \rangle$, where:
\begin{itemize}
\item $\Delta$ is a non-empty and finite set, 
\item $\nI$ is a mapping that associates each individual name with an element of $\Delta$, 
\item $\nC$ is a mapping that associates each element of $\Delta$ with a set of concepts, 
\item $\nE$ is a mapping that associates each role with a binary relation on $\Delta$.
\end{itemize}

A model graph $\langle \Delta, \nI, \nC, \nE \rangle$ is {\em consistent} and {\em $\mR$-saturated} if every $x \in \Delta$ satisfies:\footnote{A consistent and $\mR$-saturated model graph is like a Hintikka structure.} 
\begin{eqnarray}
&-\ & \textrm{$\nC(x)$ does not contain $\bot$ nor any pair $C$, $\ovl{C}$}\label{eq:HGDSX 0}\\
&-\ & \textrm{if $\tuple{x,y} \in \nE(r)$ and $r \sqsubseteq_\mR s$ then $\tuple{x,y} \in \nE(s)$}\label{eq:HGDSX b}\\
&-\ & \textrm{if $\{a\} \in \nC(x)$ then $\nI(a) = x$}\label{eq:HGDSX c}\\
&-\ & \textrm{if $C \mand D \in \nC(x)$ then $\{C,D\} \subseteq \nC(x)$}\label{eq:HGDSX 1}\\
&-\ & \textrm{if $C \mor D \in \nC(x)$ then $C \in \nC(x)$ or $D \in \nC(x)$}\label{eq:HGDSX 2}\\
&-\ & \textrm{if $\V s.C \in \nC(x)$ and $r \sqsubseteq_\mR s$ then $\V r.C \in \nC(x)$}\label{eq:HGDSX 3}\\
&-\ & \textrm{if $\tuple{x,y} \in \nE(r)$ and $\V r.C \in \nC(x)$ then $C \in \nC(y)$}\label{eq:HGDSX 4}\\
&-\ & \textrm{if $\tuple{x,y} \in \nE(r)$, $\Trans{r}$ and $\V r.C \in \nC(x)$ then $\V r.C \in \nC(y)$}\label{eq:HGDSX 4b}\\
&-\ & \textrm{if $\E r.C \in \nC(x)$ then $\E y \in \Delta$ such that $\tuple{x,y} \in \nE(r)$ and $C \in \nC(y)$}\label{eq:HGDSX 6}\\
&-\ & \textrm{if $(\geq\!n\,r.C) \in \nC(x)$ then $\sharp\{\tuple{x,y} \in \nE(r) \mid C \in \nC(y)\} \geq n$}\label{eq:HGDSX 7} \\
&-\ & \textrm{if $(\leq\!n\,r.C) \in \nC(x)$ then $\sharp\{\tuple{x,y} \in \nE(r) \mid C \in \nC(y)\} \leq n$}\label{eq:HGDSX 8} \\
&-\ & \textrm{if $(\leq\!n\,r.C) \in \nC(x)$ and $\tuple{x,y}\in \nE(r)$ then $C \in \nC(y)$ or $\ovl{C} \in \nC(y)$.}\label{eq:HGDSX 9}
\end{eqnarray}

Given a~model graph $M = \langle \Delta, \nI, \nC, \nE \rangle$, the {\em $\mR$-model corresponding to~$M$} is the interpretation $\mI = \langle \Delta, \cdot^{\mI}\rangle$ where:
\begin{itemize}
\item $a^\mI = \nI(a)$ for every individual name~$a$,
\item $A^\mI = \{x \in \Delta \mid A \in \nC(x)\}$ for every concept name $A$,
\item $r^\mI = \nE'(r)$ for every role name $r \in \RN$, where $\nE'(r)$ for $r \in \RN$ are the smallest binary relations on $\Delta$ such that:
  \begin{itemize}
  \item $\nE(r) \subseteq \nE'(r)$, 
  \item if $r \sqsubseteq_\mR s$ then $\nE'(r) \subseteq \nE'(s)$,
  \item if $\Trans{r}$ then $\nE'(r) \circ \nE'(r) \subseteq \nE'(r)$.
  \end{itemize}
\end{itemize}

Note that the smallest binary relations mentioned above always exist: for each $r \in \RN$, initialize $\nE'(r)$ with $\nE(r)$; then, while one of the above mentioned condition is not satisfied, extend the corresponding $\nE'(r)$ minimally to satisfy the condition.

\begin{lemma} \label{lemma: model graph}
If $\mI$ is the $\mR$-model corresponding to a consistent $\mR$-saturated model graph $\langle \Delta, \nI, \nC, \nE \rangle$, then $\mI$ is a model of $\mR$ and, for every $x \in \Delta$ and $C \in \nC(x)$, we have that $x \in C^\mI$.
\end{lemma}

The first assertion of this lemma clearly holds. The second assertion can be proved by induction on the structure of~$C$ in a straightforward way.

Let $G = \tuple{V,E,\nu}$ be a \CSHOQ-tableau for $\tuple{\mR,\mT,\mA}$.

Let $v \in V$ be a complex non-state with $\Status(v) \neq \Unsat$. A~{\em saturation path} of $v$ is a sequence $v_0 = v$, $v_1$, \ldots, $v_k$ of nodes of $G$, with $k \geq 1$, such that $\Type(v_k) = \State$ and 
\begin{itemize}
\item for every $1 \leq i \leq k$, $\Status(v_i) \neq \Unsat$
\item for every $0 \leq i < k$, $\Type(v_i) = \NonState$ and $\tuple{v_i,v_{i+1}} \in E$. 
\end{itemize}
Observe that each saturation path of $v$ is finite.\footnote{If a non-state $v_{i+1}$ is a successor of a non-state $v_i$ then either $\sharp\{a \in \IN \mid \Repl(v_{i+1})(a) = a\} < \sharp\{a \in \IN \mid$ $\Repl(v_i)(a) = a\}$ or $\sharp\{a \in \IN \mid \Repl(v_{i+1})(a) = a\} = \sharp\{a \in \IN \mid \Repl(v_i)(a) = a\}$ and $\FullLabel(v_{i+1}) \supset \FullLabel(v_i)$. Recall also that $\FullLabel(v_{i+1})$ is a subset of $\closure(\mR,\mT,\mA)$.\label{footnote: JHFCS}}  Furthermore, if $v_i$ is a non-state with $\Status(v_i) \neq \Unsat$ then $v_i$ has a successor $v_{i+1}$ with $\Status(v_{i+1}) \neq \Unsat$. Therefore, $v$ has at least one saturation path. 

Let $u \in V$ be a complex state and $v \in V$ be a simple non-state such that $\Status(u) \neq \Unsat$, $\Status(v) \neq \Unsat$, $\Status(v) \neq \UnsatWrt(\{u,\ldots\})$ and $v$ may affect the status of the root $\nu$ via a path through $u$. A~{\em saturation path of $v$ w.r.t.~$u$} is a sequence $v_0 = v$, $v_1$, \ldots, $v_k$ of nodes of $G$, with $k \geq 1$, such that either $\Type(v_k) = \State$ or $\Status(v_k) = \Blocked$, and 
\begin{itemize}
\item for every $1 \leq i \leq k$, $\Status(v_i)$ is not $\Unsat$ nor $\UnsatWrt(\{u,\ldots\})$,
\item for every $0 \leq i < k$, $\Type(v_i) = \NonState$ and $\tuple{v_i,v_{i+1}} \in E$. 
\end{itemize}
Observe that each saturation path of $v$ w.r.t.~$u$ is finite (see the footnote~\ref{footnote: JHFCS}). Furthermore, if $v_i$ is a non-state with $\Status(v_i)$ different from $\Unsat$ and $\UnsatWrt(\{u,\ldots\})$, then $v_i$ has a successor $v_{i+1}$ with $\Status(v_{i+1})$ different from $\Unsat$ and $\UnsatWrt(\{u,\ldots\})$. Therefore, $v$ has at least one saturation path w.r.t.~$u$. 

\begin{lemma}[Completeness of \CSHOQ] \label{lemma: completeness}
Let $G = \tuple{V,E,\nu}$ be a \CSHOQ-tableau for $\tuple{\mR,\mT,\mA}$. Suppose $\Status(\nu) \neq \Unsat$. Then $\tuple{\mR,\mT,\mA}$ is satisfiable.
\end{lemma}

\newcommand{\commentExp}[1]{($\ast$ #1 $\ast$)}
\begin{proof}
The root $\nu$ has a saturation path $u_0, \ldots,u_k$ with $u_0 = \nu$. Let $u = u_k$. We define a model graph $M = \langle \Delta, \nI, \nC, \nE\rangle$ as follows:
\begin{enumerate}
\item\label{item: JROSA} Let $\Delta_0$ be the set of all individual names $a$ such that $\Repl(u)(a) = a$. Set $\Delta := \Delta_0$ and $\nI := \Repl(u)$. If $a \in \IN$ does not occur in $\tuple{\mR,\mT,\mA}$ then define $\nI(a)$ to be some individual occurring in $\Delta_0$. 
For each $a \in \Delta_0$, mark $a$ as {\em unresolved}$\,$\footnote{Each node of $M$ will be marked either as {\em unresolved} or as {\em resolved}.} and set $\nC(a) := \{C \mid (a\!:\!C) \in \FullLabel(u)\}$.
For each role $r$, set $\nE(r) := \{\tuple{a,b} \mid r(a,b) \in \FullLabel(u)\}$.

\item\label{item: JHRES} For every {\em unresolved} node $y \in \Delta$ do:
  \begin{enumerate}
  \item If $y \in \Delta_0$ then let $v = u$ and\\ $WE = \{\tuple{w,e} \mid \tuple{v,w} \in E, e \in \ELabels(v,w)$ and $\piI(e) = y\}$.
  \item Else: 
	\begin{enumerate}
	\item Let $v = f(y)$.\\ \commentExp{$f$ is a constructed mapping that associates each node of $M$ not belonging to $\Delta_0$ with a simple state of $G$; as a maintained property of $f$, $\Status(v) \neq \Unsat$, $\Status(v) \neq \UnsatWrt(\{u,\ldots\})$ and $\nC(y)$ is the set obtained from $\FullLabel(v)$ by replacing every individual $b$ by $\Repl(u)(b)$ when $\Repl(u)(b)$ is defined.}
	\item Let $WE = \{\tuple{w,e} \mid \tuple{v,w} \in E, e \in \ELabels(v,w)\}$.
	\end{enumerate}
  \item Let $\ILCh = \ILConstraints(v) \cup \{x_{w,e} = 0 \mid$ $\tuple{w,e} \in WE$, $\piT(e) = \CQF$, $\Status(w) = \UnsatWrt(\{u,\ldots\})\}$.\\ \commentExp{$\ILCh$ is feasible because $\Status(v) \neq \Unsat$ and $\Status(v) \neq \UnsatWrt(\{u,\ldots\})$.} 

  \item Fix a solution of $\ILCh$, and for each $\tuple{w,e} \in WE$:
	\begin{enumerate}
	\item if $\piT(e) = \TUS$ then let $n_{w,e} = 1$, 
	\item else let $n_{w,e}$ be the value of $x_{w,e}$ in that solution. 
	\end{enumerate}
  \item Delete from $WE$ all the pairs $\tuple{w,e}$ with $n_{w,e} = 0$. 

  \item For each $\tuple{w_0,e} \in WE$ do:
     \begin{enumerate}
     \item Let $w_0$, \ldots, $w_h$ be a saturation path of $w_0$ w.r.t.~$u$. 
     \item Let $X$ be the set obtained from $\FullLabel(w_h)$ by replacing every individual $b$ by $\Repl(u)(b)$ when $\Repl(u)(b)$ is defined. 
     \item If $\Status(w_h) = \Blocked$ then:
	\begin{enumerate}
	\item[] \commentExp{Observe that $n_{w_0,e}$ must be 1 due to the rule \USc and the construction of $\ILConstraints(v)$ by the transitional full-expansion rule.} 
	\item Let $\{a\}$ be an element of $X$.\\ \commentExp{Observe that, due to the specification of $X$, $\Repl(u)(a) = a$ and $a \in \Delta_0$.} 
	\item For each $r \in \piR(e)$, add $\tuple{y,a}$ to $\nE(r)$.
	\end{enumerate}
     \item Else, for $i := 1$ to $n_{w_0,e}$ do:
	\begin{enumerate}
	\item Add a new element $z$ to $\Delta$ and mark $z$ as {\em unresolved}.
	\item For each $r \in \piR(e)$, add $\tuple{y,z}$ to $\nE(r)$.
	\item Set $\nC(z) := X$ and $f(z) := w_h$.\\ \commentExp{Observe that the mentioned properties of $f$ are maintained here.} 
	\end{enumerate}
     \end{enumerate}

  \item Mark $y$ as {\em resolved}.
  \end{enumerate}
\end{enumerate}

The defined model graph $M$ may be infinite. It consists of a finite base created at the step~\ref{item: JROSA} and disjoint trees created at the step~\ref{item: JHRES} possibly with edges coming back directly to $\Delta_0$ (nodes of the base) due to nominals.

It is straightforward to prove that $M$ is a consistent $\mR$-saturated model graph. 

Observe that:
\begin{itemize}
\item For any individual name $b$, if $a = \Repl(u)(b)$ then $\Repl(u)(a) = a$ and $\nI(a) = a$. 
\item If $(a\!:\!C) \in \mA$ then the concept obtained from $C$ by replacing every individual $b$ by $\Repl(u)(b)$ (when $\Repl(u)(b)$ is defined) belongs to $\nC(a')$, where $a' = \Repl(u)(a)$.
\item If $r(a,b) \in \mA$ then $\tuple{a',b'} \in \nE(r)$, where $a' = \Repl(u)(a)$ and $b' = \Repl(u)(b)$.
\item If $a \not\doteq b \in \mA$ then $a' \not\doteq b' \in \Label(u)$, where $a' = \Repl(u)(a)$ and $b' = \Repl(u)(b)$. Since $\Status(u) \neq \Unsat$, we have that $a' \neq b'$. 
\item For every $C \in \mT$, the concept obtained from $C$ by replacing every individual $b$ by $\Repl(u)(b)$ (when $\Repl(u)(b)$ is defined) belongs to $\nC(x)$ for all $x \in \Delta$.
\end{itemize}
Hence, by Lemma~\ref{lemma: model graph}, the interpretation corresponding to $M$ is a~model of $\tuple{\mR,\mT,\mA}$.
\myEnd
\end{proof}

%% file: SHOQ-llncs.bbl
\begin{thebibliography}{10}

\bibitem{dkns2011}
B.~Dunin-K\c{e}plicz, L.A. Nguyen, and A.~Sza{\l}as.
\newblock {Converse-PDL} with regular inclusion axioms: A framework for {MAS}
  logics.
\newblock {\em J. Applied Non-Classical Logics}, 21(1):61--91, 2011.

\bibitem{FaddoulH10}
J.~Faddoul and V.~Haarslev.
\newblock Algebraic tableau reasoning for the description logic {SHOQ}.
\newblock {\em J. Applied Logic}, 8(4):334--355, 2010.

\bibitem{Farsiniamarj08}
N.~Farsiniamarj.
\newblock Combining integer programming and tableau-based reasoning: a hybrid
  calculus for the description logic {SHQ}.
\newblock Master's thesis, Concordia University, 2008.

\bibitem{Gore99}
R.~Gor\'{e}.
\newblock Tableau methods for modal and temporal logics.
\newblock In D'Agostino et~al, editor, {\em Handbook of Tableau Methods}, pages
  297--396. Kluwer, 1999.

\bibitem{GoreNguyen05tab}
R.~Gor{\'e} and L.A. Nguyen.
\newblock A tableau system with automaton-labelled formulae for regular grammar
  logics.
\newblock In B.~Beckert, editor, {\em Proceedings of TABLEAUX 2005}, volume
  3702 of {\em LNAI}, pages 138--152. Springer, 2005.

\bibitem{GoreNguyenTab07}
R.~Gor\'{e} and L.A. Nguyen.
\newblock {ExpTime} tableaux with global caching for description logics with
  transitive roles, inverse roles and role hierarchies.
\newblock In N.~Olivetti, editor, {\em Proceedings of TABLEAUX 2007}, volume
  4548 of {\em LNAI}, pages 133--148. Springer, 2007.

\bibitem{GoreNguyen07clima}
R.~Gor{\'e} and L.A. Nguyen.
\newblock Analytic cut-free tableaux for regular modal logics of agent beliefs.
\newblock In F.~Sadri and K.~Satoh, editors, {\em Proceedings of CLIMA VIII},
  volume 5056 of {\em LNAI}, pages 268--287. Springer, 2008.

\bibitem{GoreN11}
R.~Gor{\'e} and L.A. Nguyen.
\newblock Exptime tableaux for {ALC} using sound global caching.
\newblock {\em J. Autom. Reasoning}, 50(4):355--381, 2013.

\bibitem{GoreW09}
R.~Gor{\'e} and F.~Widmann.
\newblock Sound global state caching for $\mathcal{ALC}$ with inverse roles.
\newblock In M.~Giese and A.~Waaler, editors, {\em Proceedings of TABLEAUX
  2009}, volume 5607 of {\em LNCS}, pages 205--219. Springer, 2009.

\bibitem{GoreW10}
R.~Gor{\'e} and F.~Widmann.
\newblock Optimal and cut-free tableaux for propositional dynamic logic with
  converse.
\newblock In J.~Giesl and R.~H{\"a}hnle, editors, {\em Proceedings of IJCAR
  2010}, volume 6173 of {\em LNCS}, pages 225--239. Springer, 2010.

\bibitem{HladikM04}
J.~Hladik and J.~Model.
\newblock Tableau systems for {SHIO} and {SHIQ}.
\newblock In {\em Proceedings of Description Logics'2004}, volume 104 of {\em
  CEUR Workshop Proceedings}, pages 168--177, 2004.

\bibitem{HorrocksKS06}
I.~Horrocks, O.~Kutz, and U.~Sattler.
\newblock The even more irresistible $\mathcal{SROIQ}$.
\newblock In P.~Doherty, J.~Mylopoulos, and C.A. Welty, editors, {\em Proc.\ of
  KR'2006}, pages 57--67. AAAI Press, 2006.

\bibitem{HorrocksS01}
I.~Horrocks and U.~Sattler.
\newblock Ontology reasoning in the {SHOQ(D)} description logic.
\newblock In B.~Nebel, editor, {\em Proceedings of IJCAI'2001}, pages 199--204.
  Morgan Kaufmann, 2001.

\bibitem{BranchAndBound}
A.~H. Land and A.~G. Doig.
\newblock An automatic method of solving discrete programming problems.
\newblock {\em Econometrica}, 28(3):497--520, 1960.

\bibitem{nguyen01B5SL}
L.A. Nguyen.
\newblock Analytic tableau systems and interpolation for the modal logics {KB},
  {KDB}, {K5}, {KD5}.
\newblock {\em Studia Logica}, 69(1):41--57, 2001.

\bibitem{Nguyen-ALCI}
L.A. Nguyen.
\newblock Cut-free {ExpTime} tableaux for checking satisfiability of a
  knowledge base in the description logic $\mathcal{ALCI}$.
\newblock In {\em Proceedings of ISMIS'2011}, volume 6804 of {\em LNAI}, pages
  465--475. Springer, 2011.

\bibitem{SHI-ICCCI}
L.A. Nguyen.
\newblock A cut-free {ExpTime} tableau decision procedure for the description
  logic {SHI}.
\newblock In {\em Proceedings of ICCCI'2011 (1)}, volume 6922 of {\em LNCS},
  pages 572--581. Springer, 2011 (see also the long version
  \url{http://arxiv.org/abs/1106.2305}).

\bibitem{SHIQ-long}
L.A. Nguyen.
\newblock {ExpTime} tableaux for the description logic {SHIQ} based on global
  state caching and integer linear feasibility checking.
\newblock arXiv:1205.5838, 2012.

\bibitem{SHIQ}
L.A. Nguyen.
\newblock A tableau method with optimal complexity for deciding the description
  logic {SHIQ}.
\newblock In {\em Proceedings of ICCSAMA'2013}, volume 479 of {\em Studies in
  Computational Intelligence}, pages 331--342. Springer, 2013.

\bibitem{SHOQ-CSP}
L.A. Nguyen and J.~Goli{\'n}ska-Pilarek.
\newblock An exptime tableau method for dealing with nominals and quantified
  number restrictions in deciding the description logic shoq.
\newblock In {\em Proceedings of CS\&P'2013}.

\bibitem{NguyenSzalas09ICCCI}
L.A. Nguyen and A.~Sza{\l}as.
\newblock {ExpTime} tableaux for checking satisfiability of a knowledge base in
  the description logic {ALC}.
\newblock In N.T. Nguyen, R.~Kowalczyk, and S.-M. Chen, editors, {\em
  Proceedings of ICCCI'2009}, volume 5796 of {\em LNAI}, pages 437--448.
  Springer, 2009.

\bibitem{NguyenSzalas-KSE09}
L.A. Nguyen and A.~Sza\l{}as.
\newblock An optimal tableau decision procedure for {Converse-PDL}.
\newblock In N.-T. Nguyen, T.-D. Bui, E.~Szczerbicki, and N.-B. Nguyen,
  editors, {\em Proceedings of KSE'2009}, pages 207--214. IEEE Computer
  Society, 2009.

\bibitem{NguyenS10FI}
L.A. Nguyen and A.~Sza\l{}as.
\newblock Checking consistency of an {ABox} w.r.t. global assumptions in {PDL}.
\newblock {\em Fundamenta Informaticae}, 102(1):97--113, 2010.

\bibitem{NguyenS10TCCI}
L.A. Nguyen and A.~Sza\l{}as.
\newblock Tableaux with global caching for checking satisfiability of a
  knowledge base in the description logic $\mathcal{SH}$.
\newblock {\em T. Computational Collective Intelligence}, 1:21--38, 2010.

\bibitem{NguyenS11SL}
L.A. Nguyen and A.~Sza\l{}as.
\newblock {ExpTime} tableau decision procedures for regular grammar logics with
  converse.
\newblock {\em Studia Logica}, 98(3):387--428, 2011.

\bibitem{PanH02}
J.Z. Pan and I.~Horrocks.
\newblock Reasoning in the {SHOQ(D$_n$)} description logic.
\newblock In {\em Proceedings of DL'2002}, volume~53 of {\em CEUR Workshop
  Proceedings}, pages 53--62, 2002.

\bibitem{Pratt80}
V.R. Pratt.
\newblock A near-optimal method for reasoning about action.
\newblock {\em J.~Comp.~Syst.~Sci.}, 20(2):231--254, 1980.

\bibitem{Rautenberg83}
W.~Rautenberg.
\newblock Modal tableau calculi and interpolation.
\newblock {\em JPL}, 12:403--423, 1983.

\bibitem{TobiesThesis}
S.~Tobies.
\newblock {\em Complexity results and practical algorithms for logics in
  knowledge representation}.
\newblock PhD thesis, RWTH-Aachen, 2001.

\bibitem{DLnavigator}
\url{http://owl.cs.manchester.ac.uk/navigator/}.

\end{thebibliography}
